\documentclass[11pt]{article}
\usepackage{amsthm}
\usepackage{amsmath}
\usepackage{amsfonts}
\usepackage{xcolor}
\usepackage{fullpage}

\newtheorem{definition}{Definition}
\newtheorem{theorem}{Theorem}
\newtheorem{conjecture}{Conjecture}

\newtheorem{lemma}{Lemma}
\newtheorem{proposition}{Proposition}
\newtheorem{corollary}{Corollary}
\newtheorem{algorithm}{Algorithm}
\newtheorem*{remark}{Remark}
\newcommand{\N}{{\mathbb N}}
\newcommand{\C}{{\mathbb C}}
\DeclareMathOperator{\parent}{par}
\DeclareMathOperator{\argmax}{argmax}
\DeclareMathOperator{\BP}{BP}

\DeclareMathOperator{\path}{path}

\title{Parallels Between Phase Transitions and Circuit Complexity?}
\author{Ankur Moitra\thanks{Department of Mathematics, Massachusetts Institute of Technology. Email: {\tt moitra@mit.edu}. This work was supported in part by NSF CAREER Award CCF-1453261, NSF Large CCF-1565235, a David and Lucile Packard Fellowship, an Alfred P. Sloan Fellowship and an ONR Young Investigator Award.} \\ MIT 
\and Elchanan Mossel\thanks{Department of Mathematics, Massachusetts Institute of Technology. Email: {\tt elmos@mit.edu}. This work was supported in part by NSF awards DMS-1737944 and 1665252, ONR Award N00014-17-1-2598, Simons Investigator in Mathematics award (622132))}  \\ MIT \and Colin Sandon  \\ MIT}

\begin{document}
\maketitle

\begin{abstract}
In many natural average-case problems, there are or there are believed to be critical values in the parameter space where the structure of the space of solutions changes in a fundamental way. These phase transitions are often believed to coincide with drastic changes in the computational complexity of the associated problem. 

In this work, we study the circuit complexity of inference in the broadcast tree model, which has important applications in phylogenetic reconstruction and close connections to community detection. We establish a number of qualitative connections between phase transitions and circuit complexity in this model. Specifically we show that there is a $\mathbf{TC}^0$ circuit that competes with the Bayes optimal predictor in some range of parameters above the Kesten-Stigum bound. We also show that there is a $16$ label broadcast tree model beneath the Kesten-Stigum bound in which it is possible to accurately guess the label of the root, but beating random guessing is $\mathbf{NC}^1$-hard on average. The key to locating phase transitions is often to study some intrinsic notions of complexity associated with belief propagation \--- e.g. where do linear statistics fail, or when is the posterior sensitive to noise? Ours is the first work to study the complexity of belief propagation in a way that is grounded in circuit complexity. 

 %   Belief propagation is one of the foundations of probabilistic and causal reasoning. In this paper, we study the circuit complexity of some of the various tasks it performs. Specifically, in the broadcast tree model (which has important applications to phylogenetic reconstruction and close connections to community detection), we show the following:

%\begin{enumerate}

%\item[(1)] No $\mathbf{AC}^0$ circuit can guess the label of the root with positive advantage over random guessing, independent of the depth for any non-trivial choice of parameters.

%\item[(2)] There is a $\mathbf{TC}^0$ circuit that competes with the Bayes optimal predictor in some range of parameters above the Kesten-Stigum bound

%\item[(3)] There is a $16$ label broadcast tree model in which it is possible to accurately guess the label of the root, but beating random guessing is $\mathbf{NC}^1$-hard

%\end{enumerate}

%\noindent Our work yields a simple and natural generative model where large depth really is necessary for performing various types of inference, that have intriguing parallels with phase transitions from statistical physics. 
\end{abstract}

%\tableofcontents
%\newpage

%Consider a tree of depth $d$ in which each vertex has $k$ children, the root is randomly assigned a community in $\{0,1\}$, and each other vertex is independently assigned its parent's community with probability $1-\epsilon$ and the opposite community with probability $\epsilon$. Replacing $\epsilon$ with $1-\epsilon$ simply inverts the communities of every node at an odd depth, so we will assume without loss of generality that $\epsilon\le 1/2$. The question is, what are the computational requirements to determine the community of the root from the communities of the leaves. To be more formal about this, let $v^{(d')}_1,\cdots,v^{(d')}_{k^{(d')}}$ be the vertices at depth $d'$, ordered so that the children of $v^{(d'-1)}_i$ are $v^{(d')}_{ki-(k-1)},\cdots,v^{(d')}_{ki}$. Also, let $X^{(d')}_i$ be the community of $v^{(d')}_i$ for each $d'$ and $i$. Finally, let $n=k^d$. When we use asymptotics, we will standardly be assuming that $k$ and $\epsilon$ stay constant while $d$ increases.
\thispagestyle{empty}

\newpage

\setcounter{page}{1}

\section{Introduction}

\subsection{Background}

In many basic problems in high-dimensional statistics and machine learning, there appear to be fundamental gaps between the performance of the information-theoretically best estimator and the best estimator that can be computed in polynomial time. These are called {\em computational vs. statistical tradeoffs}. Recently, there has been an effort 
to study these gaps in a systematic fashion, in particular by forging reductions between some of these problems. 
%There are some cases where we can give evidence by finding a reduction from one problem to another. 
For example, finding sparse directions with large variance in the spiked covariance model turns out to be at least as hard as finding small planted cliques, see e.g.~\cite{BerthetRigollet:13,MaWu:15,BrBrHu:18}.  
%In general, we are interested in ways to give evidence that these problems really are computationally hard beyond just the fact that we are currently unable to close these gaps. 
%There are some cases where we can give evidence by finding a reduction from one problem to another. For example, finding sparse directions with large variance in the spiked covariance model turns out to be at least as hard as finding small planted cliques. 
However, these reductions leave much to be desired as there are relatively few examples where reductions are known that map natural distributions on one problem to natural distributions on another.

In this paper, we will explore other popular methodologies for predicting where average-case problems become hard, which come from statistical physics and revolve around a powerful algorithm called belief propagation. Our key example originates from the following special case of community detection in the stochastic block model. We start with a fixed partition of $n$ nodes into $q$ (almost) equal sized communities. The probability of connecting any pair of nodes with an edge is 
%$\frac{a}{n}$ 
$k q \theta /n + k(1-\theta)/n$  
if they belong to the same community and otherwise is 
$k(1-\theta)/n$,
%$\frac{b}{n}$. 
where edges in the graph are sampled independently. It is easy to see that the average degree in this graph is $k$ and that $\theta$ is a measure of the strength of the communities. 

The  goal is, given a graph sampled from this model, to find a $q$-partition of its nodes whose parts have non-trivial correlation (i.e. better than random) with the true communities. A striking prediction from statistical physics~\cite{DKMZ:11} is that the problem is efficiently solvable when 
%$(a-b)^2 > k (a + (k-1) b)$, 
$k \theta^2 > 1$ 
while the information theory threshold for the problem is different for large values of $k$. 
By now the existence of efficient algorithms when $k\theta^2 > 1$ has been established \cite{MoNeSl:15,Massoulie:14,MoNeSl:18,BoLeMa:15,AbbeSandon:15} as well as the fact that for $k > 5$, the information theory threshold is strictly below this bound \cite{AbbeSandon:15,BMNN:16}. 

The threshold of 
%$(a-b)^2 > k (a + (k-1) b)$ 
$k \theta^2 > 1$ 
is called the {\em  Kesten-Stigum bound} and will play an important role in our paper. It is believed that for some problems, like the block model, the structure of the space of solutions changes in a fundamental way beneath the Kesten-Stigum bound, and this is the basis for the predictions about computational hardness. %The connection between the broadcast model on the tree and computational hardness has been extensively studied at the interface of computing, statistical physics and mathematics, see~\cite{MezardMontanari:06,KMRSZ:07} where the Kesten-Stigum transition as well as other transitions for tree broadcast processes are used to predict statistical/computational gaps.  
Fundamentally, these predictions of computational difficulty all revolve around studying the behavior of belief propagation. In what follows we will explain some of the intuition behind belief propagation along with how do computational versus statistical phase transitions are predicted. See also \cite{MezardMontanari:06,KMRSZ:07}. 

The way to think about belief propagation in the stochastic block model is to start with a local view around a node. With high probability, its neighborhood will be tree-like. In fact, we can model it (along with which community each node belongs to) as a Markov process on a tree. This model is called the {\em broadcast tree model}. We start with a complete $k$-regular tree of height $d$ (or alternatively we generate a random tree of height $h$ in which the number of children of each node is a Poisson random variable with expectation $k$). The root is assigned one of the $q$ possible labels at random. Next we propagate labels from the root to the leaves by, at each step, assigning a child the same label as its parent with probability $\theta$ and otherwise assigning it a uniformly random label. At the end, we are given the labels of the leaves and the goal is to use this information to guess the label of the root. We want our guess to be correct with some advantage over random guessing, and we want  the advantage to be bounded away from zero independently of $h$. Belief propagation is an iterative algorithm that provably computes the posterior distribution on the label of the root given the labels of the leaves. So when belief propagation fails at guessing the label of the root with some nonzero advantage that is independent of $h$, it is because the problem is information-theoretically impossible. Belief propagation is based  on the idea that conditioned on the label of some node, the labels of its neighbors are independent. This is exactly true on a tree and approximately  true in a sparse random graph with few short cycles. 

\begin{quote} {\em The key to using belief propagation to locate phase transitions is that it has its own intrinsic notions of complexity.}
\end{quote}

In the broadcast tree model, the Kesten-Stigum bound is the threshold $k \theta^2  > 1$. (The  Kesten-Stigum bound in the stochastic block model is usually stated in terms of $a$ and $b$ 
%and the one in the broadcast tree model may look different, 
but they are actually the same, which can be seen by relating $a, b, \theta$ and $k$). It turns out that the Kesten-Stigum bound coincides with where linear statistics stop working.  In fact, in the seminal work of Kesten and Stigum \cite{KestenStigum:67,KestenStigum:66}, 
they showed that it is possible to guess the label of the root (and beat random guessing) just by tallying the number of labels of each type among the leaves. Moreover, it is not too hard to deduce from their results~\cite{MosselPeres:03} that below the Kesten-Stigum bound, this method fails. 
Perhaps surprisingly, it is still possible to  guess the label of the root and beat random guessing beneath the Kesten-Stigum bound when $k \geq 5$. However, this requires to use {\em higher-order} information about which labels appear where in the tree~\cite{Mossel:01,Sly:09,Sly:09a}. 

Alternatively, the Kesten-Stigum bound can be thought of through the lens of robustness. Suppose we inject random noise at the leaves. In particular, suppose we overwrite the label of each leaf to a random value with probability $\eta$. Then above the Kesten-Stigum bound, reconstructing the root in the face of noise is still possible, but beneath the Kesten-Stigum bound it is information-theoretically impossible~\cite{JansonMossel:04}. Thus the Kesten-Stigum bound is the location in parameter space where the typical posterior distribution on the label of the root becomes highly sensitive to noise. 

%Finally, there are other approaches for locating phase transitions based on the cavity method, which  studies the stability of the trivial fixed point of belief propagation, in order to predict where computational hardness (in the stochastic block model) sets in. You can even use the closely related replica method, or identify barriers in the energy landscape in message passing. 

Fundamentally, each of these methodologies represents a way to extract information from belief propagation about where the posterior distribution on the label of the root becomes highly complex. The notion of complexity is expressed in many different ways \--- for example, the failure of linear statistics, lack of robustness, or (in the physics language) stability of the trivial fixed point. In this paper, we take an approach that is grounded in computational complexity for studying the posterior distribution in the broadcast tree model. (Alternatively, we take a circuit complexity approach to studying the complexity of the problem that belief propagation is actually  solving). 

\begin{quote} {\em We establish some tantalizing parallels between phase transitions (in the traditional meaning of the phrase, where it refers to changes in the structure of the solution space) and phase transitions in the circuit complexity of the inference problem. }
\end{quote}

\subsection{Our Results}

% detection, where the goal is to guess the label of the root, given leaves generated at random,
In this paper, we study the circuit complexity of various tasks performed by belief propagation on the broadcast tree model. We will be interested in four main problems: $(1)$ detection, where the goal is to guess the label of the root, given leaves generated at random, with probability $1/q + \epsilon$ with $\epsilon > 0$ independent of the depth $(2)$ inference, where the goal is to compete with the Bayes optimal predictor asymptotically in an average-case sense over samples from the model $(3)$ computing the posterior, which is the analogous question for worst-case inputs on the labels of the leaves. And finally we study $(4)$ the complexity of the forward problem of generating samples from the model. These tasks can all naturally be solved in $\mathbf{NC}^1$ the class of logarithmic depth circuits with AND, OR and NOT gates. However it will turn out that in some cases (conjecturally) weaker classes with constant depth will suffice and in others logarithmic depth is inherently necessary. 

It is well known that for the broadcast tree model on two labels \--- also called the Ising model on trees \--- beneath the Kesten-Stigum bound detection is information-theoretically impossible. What this means is that taking the majority vote of the labels of the leaves solves the detection problem whenever it is information-theoretically possible to do so. However it is also well-known that majority vote is suboptimal in how often it guesses the label of the root correctly. Intuitively, this is because there is more information about the label of the root contained not just in the number of labels of each type but also in the structure of where in the tree they are relative to each other. We prove that there are more complex circuits, but still ones in $\mathbf{TC}^0$, that can solve the inference problem:

\begin{theorem}[informal, see Theorem~\ref{thm:tc0main}]
There is a constant $C > 1$ so that  $k \theta^2  > C$ then the inference problem in the Ising model ($q=2$) on trees can be solved in $\mathbf{TC}^0$. 
\end{theorem}

Our approach is based on~\cite{MoNeSl:14b} that shows belief propagation (suitably above the Kesten-Stigum bound) is robust to label noise. This allows to construct a $\mathbf{TC}^0$ circuit by using majority on the leaves of a subtree to get noisy estimates of their roots. We then bootstrap these estimates to get asymptotically optimal estimates of the label of the overall root. It is conjectured that belief propagation works with noisy labels all the way down to the Kesten-Stigum bound (i.e. $k \theta^2  > 1$) in which case we could improve the above theorem analogously. 

As we discussed earlier, belief propagation works even in a worst-case sense and computes the true posterior. We show that the worst-case problem is much harder and is $\mathbf{NC}^1$-complete:

\begin{theorem}[informal, see Theorem~\ref{thm:nc1main}]
There are constants $\theta$ and $k$ for which computing the posterior in the Ising model on trees is $\mathbf{NC}^1$-complete. 
\end{theorem}

%This shows that even for two labels where there is no phase transition in the traditional sense (i.e. a gap in parameters where majority vote fails but belief propagation works for detection) there can be instances where the posterior becomes a more complex function of the labels of the leaves. 

However there is something unsatisfying about a circuit complexity lower bound that applies to the problem of computing the posterior distribution on the label of the root for a worst-case configuration of labels on the leaves. The broadcast tree model is a generative model, and the properties of belief propagation that are used to locate phase transitions are really average-case properties \--- or rather, properties about the posterior distribution on the label of the root, for a typical realization of the labels of the leaves. Now we come to what we believe to be our most significant result. We study the average-case circuit complexity of guessing the label of the root in a broadcast tree model whose parameters are beneath the Kesten-Stigum bound. We prove:

\begin{theorem}[informal, see Theorem~\ref{thm:nc1main2}]
There is a $16$ label broadcast tree model where it is possible to guess the label of the root with probability $\geq 0.999$ but where detection is $\mathbf{NC}^1$-complete. 
\end{theorem}

For a general markov process on a $k$-regular tree with a transmission matrix $M$, the Kesten-Stigum bound is $k (\lambda_2(M))^2 > 1$ where $\lambda_2(M)$ is the second largest eigenvalue of $M$. In our construction, the transmission matrix has a second eigenvalue equal to zero and thus no matter how large $k$ is, we are operating below the Kesten-Stigum bound. (Equivalently, no matter how large $k$ is, linear statistics are not enough to guess the label of the root with positive advantage over random guessing). More broadly, we conjecture that the detection problem is $\mathbf{NC}^1$-complete {\em anywhere} beneath the Kesten-Stigum bound, which is consistent with the fractal way that information is stored in such settings \cite{Mossel:01}, but we are only able to prove it for this particular $16$ label broadcast tree model.

Barrington famously showed that the word problem over nonsolvable groups is $\mathbf{NC}^1$-complete \cite{barrington1989bounded}. This leads to a natural average-case $\mathbf{NC}^1$-complete problem via telescopically multiplying by random group elements. We construct a model where the labels of the children can be multiplied to get the labels of the parents. While we can solve detection by multiplying group elements in some way, what is less obvious is how to show that any circuit for detection can be used to solve the word problem. The key idea is we can define an alternative but equivalent generation procedure that starts by labelling the root implicitly as the product of many group elements, and as we follow the process down the levels of the tree, the product simplifies and involves fewer elements until at the leaves it is a random function of a single group element. In this way, the generative process expresses the label of the root as a random function of the labels of the leaves, as opposed to the other way around. This is our most challenging result and perhaps the most surprising.

Finally, we study the circuit complexity of some of the remaining tasks associated with the broadcast tree model to complete the picture. First, it is natural to wonder if weaker circuit models can solve the detection problem. We show an unconditional lower bound against $\mathbf{AC}^0$:

\begin{theorem}[informal, see Theorem~\ref{thm:ac0main}]
For any $0 < \theta < 1$, there is no $\mathbf{AC}^0$ circuit for solving the detection problem in the Ising model on trees. 
\end{theorem}

\noindent The proof is based on the observation that the generative process for the broadcast tree model can itself be thought of as a random restriction \---- a classic tool for proving circuit lower bounds \cite{furst1984parity}. The main difference is that we do not get to choose the parameters of the restriction ourselves, it is dictated by the model and only sets a constant fraction of the inputs as we go up one level of the tree. Luckily, we can define an alternative generative process that is equivalent to the broadcast tree model but uses random restrictions as an intermediary step. 

Despite the fact that $\mathbf{AC}^0$ circuits do not solve even the most basic type of inference problem in any interesting range of parameters, it turns out that, somewhat surprisingly, they can solve the forward problem of generation. 

\begin{theorem}[informal, see Theorem~\ref{thm:genmain}]
For any $\theta = a/2^b$ where $a$ and $b$ are integers, given uniformly random bits as input, there is an $\mathbf{AC}^0$ circuit for sampling from the Ising model on trees. 
\end{theorem}

Thus the broadcast tree model on two labels is an interesting example where there is a wide discrepancy between the depth needed for generation vs. inference. This is reminiscent of the work of Babai \cite{babai1987random} and Boppana and Lagarias \cite{boppana1987one} who show that, while $\mathbf{AC}^0$ cannot compute parity on the uniform distribution, there is a depth one circuit whose outputs depend on two bits each that samples from a distribution whose first $n$ bits are uniform and whose last bit is their parity. 

%Recently, a variety of applied works have suggested particular deep generative models that seem to need comparable levels of depth for inference. Our work shows that super-constant depth is needed even when the generative process is simpler and can be solved by a constant depth circuit. Our work has its intellectual roots in the work of Mossel that uses the broadcast tree model to understand where inference provably requires decoding higher-order correlations, and makes some interesting connections to deep learning. 

%\note{Multiplying the labels of a child gives part of the label of its parent, but you can not tell which part easily. As a result, I do not see any natural way to solve detection by solving a word problem over the group.}

\subsection{More Related Work}

We note that while our depth lower bounds result apply to a natural inference problem, the results proving logarithmic lower bounds are conditional (on the fact that $\mathbf{NC}^1 \neq \mathbf{TC}^0$).
This should be compared to the unconditional lower bounds for deep nets~\cite{Telgarsky:16} and to worst case~\cite{Hastad:87} and average case~\cite{HRRT:17} lower bounds in circuit complexity. 
In fact, part of the motivation for our work comes from the work of the second author~\cite{Mossel:19deep} who suggested that the broadcast model is a particularly natural data generative model that has provable reconstruction algorithms and for which one can prove rigorously that depth is needed for inference. The reconstruction algorithms of the broadcast process, are often referred to as phylogenetic reconstruction algorithm. Polynomial time algorithm for reconstructing phylogenies were established in~\cite{ErStSzWa:99a,ErStSzWa:99b} and phase transition related to the Kesten-Stigum bound in the model were established in~\cite{Mossel:03,Mossel:04a} and follow up work. The paper~\cite{Mossel:19deep} does not prove depth lower bounds in the sense of the current paper. Rather, it shows that for a range of values of $\theta$,  in a semi-supervised broadcast setting, algorithms that can only access low-order moments of the labelled data are unable to classify better than random, while there exists algorithms that use high-order moments and are able to label accurately.
In a concurrent work~\cite{JKLM:19}, it was shown that message passing algorithms that use only bounded memory of bits per node, do not achieve the Kesten-Stigum Bound even for the Ising Model ($q=2$). This proves a conjecture from~\cite{EvKePeSc:00}. However, these results do not have implications for the circuit complexity of the problem. 

There is also a close connection between the types of problems we study here and the coin problem in pseudorandomness \cite{brody2010coin}, which asks: Suppose we are given a coin which is promised to have bias either $1/2 + \delta$ or $1/2 -\delta$ along with $n$ independent tosses and our goal is to guess which way the coin is biased and to guess correctly with (say) probability at least $2/3$. What is the smallest $\delta$ for which a given computational model (e.g. $\mathbf{AC}^0$ \cite{shaltiel2010hardness, aaronson2009bqp}, width $w$ ROBPs \cite{brody2010coin}) can succeed? In fact we can think of this as a broadcast problem on a $n$-ary depth one tree with two labels where the label of the root represents whether the coin has positive or negative bias. 

With an unrestricted computational model, the majority function is optimal. And thus the coin problem is interesting in models that cannot compute the majority function and in turn leads to bounds on the fourier coefficients of the functions that they can compute and is a key ingredient in various PRGs. In the broadcast tree model, the labels of the leaves are no longer independent conditioned on the root but rather have a hierarchical structure to the strength of their dependencies. As it turns out, in light of our results, this problem can be much harder. We show that it is $\mathbf{NC}^1$-complete for a particular broadcast problem on $16$ labels. Optimistically, and in analogy with the coin problem, we could ask: Could proving unconditional lower bounds against $\mathbf{TC}^0$ for the broadcast tree problem lead to non-trivial PRGs?

\section{Preliminaries}

%\fbox{\parbox{16cm}{Wikipedia puts the numbers for a circuit class in a superscript, like $\mathbf{AC}^0$. Also, we could use $\mathbb{P}$ for probabilities.}}

\subsection{The broadcast tree model}

In this paper we consider the classical tree broadcast model on 
regular trees and binary labels. Throughout we will use the following notation. We write $T_k(d)$ for the $d$-level $k$-ary tree.   We will identify such a tree $T_k(d)$ with a
subset of $\N^*$, the set of finite strings of natural numbers,
with the property that if
$v \in T$ then any prefix of $v$ is also in $T$. In this way, the root of the
tree is naturally identified with the empty string, which we will denote by $\rho$.
We will write $uv$ for the concatenation of the strings $u$ and $v$, and
$L_r(u)$ for the $r$th-level descendents of $u$; that is,
$L_r(u) = \{uv \in T: |v| = r\}$. Also, we will write
$\C(u) \subset \N$ for the indices of $u$'s children relative to itself.
That is, $i \in \C(u)$ if and only if $ui \in L_1(u)$. We write $L_r$ for $L_r(\rho)$ and $\parent(v)$ 
for the parent of node $v$. 

\begin{definition}[Broadcast process on a tree]
Given a parameter $\theta \in [-1, 1]$ and a $k$-ary tree of $d$ level $T_k(d)$, the {\em broadcast process on  $T$} is
 a two-state Markov process $\{\sigma_u : u \in T\}$ defined as follows:
let $\sigma_\rho$ be $1$ or $0$ with probability $\frac{1}{2}$. Then, for each $u$ such that $\sigma_u$ is defined,
independently for every $v \in L_1(u)$ let $\sigma_v = \sigma_u$ with probability
$\theta + (1-\theta)/2$ and $\sigma_v = 1-\sigma_u$ otherwise.
\end{definition}

In other words, in the broadcast model, the root is randomly assigned a label in $\{0,1\}$, and then each other vertex is assigned its parent's label with probability $\theta$ and an independent uniformly chosen label with probability $1-\theta$.
Of course, this is equivalent to keeping the bit with probability $1/2 + \theta/2$ and flipping it to the opposive value with 
probability $1/2 - \theta/2$.

This broadcast process has been extensively studied in probability, where the major question is
whether the labels of vertices far from the root of the tree give
any information on the label of the root,~\cite{KestenStigum:66,BlRuZa:95}. See also~\cite{EvKePeSc:00,Mossel:04,MezardMontanari:06}.
A similar question was studied in various communities including bio-informatics~\cite{Felsenstein:04} and AI~\cite{Pearl88} from an algorithmic perspective, where the goal is to estimate (the posterior) of the root given the labels of vertices far from the root. It is well known that Belief Propagation is an exact linear time algorithm for computing the posterior.

We will mainly be focusing on the asymptotic behavior of the broadcast model as $d$ increases with all other parameters held constant, and we will commonly set $n=k^d$. We will be discussing the circuit complexity of multiple tasks associated with the broadcast model on the tree. To simplify notation we write $X^{(r)}$ for the vector of labels at level $r$:
$ X^{(r)} := (\sigma_v : |v| = r) $. 

The most important task associated with the model is inference of the root given $X^{(d)}$. As mentioned earlier, Belief propagation is used for this task. The output of Belief propagation is a posterior distribution $\mathbb{P}[X^{(0)} = \dot | X^{(d)} = x]$. For a fixed $d$ and $k$ the posterior is always bounded away from $0$ and $1$. Indeed if $k$ is even, the posterior can often assign equal probability to the two root values.
Rounding the posterior allows to determine the more likely root value. The probabilistic nature of the inference problem, leads to a number of complexity formulations. First, in the worst-case formulation, we are looking for circuits that estimate the root correctly whenever the posterior is far enough from $(1/2,1/2)$. 
In terms of average case, there is a natural distribution over the inputs, i.e, the distribution given by the broadcast process. It is thus natural to formulate an average case version of the problem where the inputs are drawn from the broadcast distribution and the objective is to estimate the root correctly with almost the same probability that BP does. Finally, in the average case setup we may settle for less, i.e., inferring the root correctly with probability bounded away from $1/2$. 
The formal definition of the $3$ problems follow.

%Since our main interest is in analyzing classical complexity measures for inference tasks, we are mostly interested in decision problems. 

%drawing a sample from a given instance of the tree broadcast model. However, the main focus of this paper will be the question of how hard it is to determine $X^{(0)}$ from $X^{(d)}$ when $X$ is drawn from the TBM with a given set of parameters. Of course, it is ambiguous what exactly we require an algorithm to do in order to be considered a success at this. We could require it to determine $X^{(0)}$ with accuracy $1/m+\Omega(1)$, determine $X^{(0)}$ with an accuracy that is at least as good as any other algorithm, or even require it to give a reasonable guess of the value of $X^{(0)}$ for every possible value of $X^{(d)}$. As such, we define the following.

%\fbox{\parbox{16 cm}{\em I still do not like posterior functions as a term. Maybe we could call them conclusive functions or something.}}

\begin{definition}
We say that a series of functions $f_d:\{0,1\}^{L_d}\to\{0,1\}$ are {\em posterior} functions if 
%recovers the root of $TBM(m,k,M)$ with strongly optimal accuracy 
%if there exists $\delta_d=o(1)$ such that if %$(X^{(0)},\cdots,X^{(d)})\sim TBM_d(m,k,M)$ then
\[
\mathbb{P}[X^{(0)}=f(x)|X^{(d)}=x]\ge \mathbb{P}[X^{(0)}=\BP(x)|X^{(d)}=x]-\delta_d\]
for every $d$ and every $x\in\{0,1\}^{L_d}$, where 
$\BP(x) := \argmax_{a \in \{0,1\}} P[X^{0} = a | X^{d} = x]$ is the optimal Bayes posterior, i.e., the one obtained by applying Belief Propagation and rounding, 
and $\delta_d \to 0$ as $d \to \infty$. 
\end{definition}
%We note that $\BP(x) := \argmax_{a \in \{0,1\}} P[X^{0} = a | X^{d} = x]$ is the optimal Bayes posterior, i.e., the one obtained by applying Belief Propagation and rounding. 

\begin{definition}
We say that a series of functions $f_d:\{0,1\}^{L_d}\to\{0,1\}$ are {\em inference} functions if 
%A series of functions $f_d:\{0,\cdots,d-1\}^{k^d}\to\{0,\cdots,d-1\}$ recovers the root of $TBM(m,k,M)$ with optimal accuracy if for every series of functions $f'_d:\{0,\cdots,d-1\}^{k^d}\to\{0,\cdots,d-1\}$, in the limit as $d\to\infty$, if $(X^{(0)},\cdots,X^{(d)})\sim TBM_d(m,k,M)$ then
\[
\mathbb{P}[f(X^{(d)})=X^{(0)}]\ge \mathbb{P}[\BP(X^{(d)})=X^{(0)}]-\delta_d,
\]
where $\delta_d \to 0$ as $d \to \infty$
\end{definition}
Thus a function is an inference function if they find the most likely root with (almost) the same overall probability as Belief Propagation does. 

\begin{definition}
We say that a series of functions $f_d:\{0,1\}^{L_d}\to\{0,1\}$ are {\em detection} functions if there exists $\delta>0$ and $d_0$ such that for all $d\ge d_0$,
%if $(X^{(0)},\cdots,X^{(d)})\sim TBM_d(m,k,M)$ then
\[
\mathbb{P}[f(X^{(d)})=X^{(0)}]\ge 1/2+\delta
\]
\end{definition}

In other words, a series of detection functions determines the root's label with accuracy $1/2+\Omega(1)$, a series of inference functions determines the root's label with an accuracy within $o(1)$ of the best possible, and a series of posterior functions 
 determines the root's label with an accuracy within $o(1)$ of the best possible conditioned on any possible value of $X^{(d)}$. Clearly posterior functions are also inference functions.  
When the reconstruction problem is unsolvable, there are no detection functions. If it is solvable, then inference functions are also detection functions. 

In addition to inference problem, we are also interested in the generation problem, in other words, what is the computation complexity of generating $X^{(d)}$ given access to random bits. We address the generation question is section~\ref{sec:generation}
%If any series of functions recovers the root with nontrivial accuracy, then anything that recovers it with optimal accuracy also recovers it with nontrivial accuracy, because it must be at least as good as that series. If no series of functions recovers it with nontrivial accuracy, then every series of functions recovers it with optimal accuracy, because the optimal accuracy is $1/m$.

%Our conclusions about the generation problem is that there are $\mathbf{AC}^0$ circuits that can generate $X^{(d)}$ but that no $\mathbf{NC}^0$ circuit can generate $X^{(d)}$ unless $\theta$ is $0$ or $\pm 1$. Our main results on the inference problem are that $\mathbf{AC}^0$ circuits are ineffective at determining the root's label, $\mathbf{TC}^0$ circuits can be somewhat effective at determining the root's label, and that $\mathbf{NC}^1$ circuits can be highly effective at recovering the root's label. More formally, we have the following.

%\begin{theorem}
%If $0\le \theta<1$ and $f$ is a detection function, then $f$ is not in $\mathbf{AC}^0$.
%\end{theorem}

%\begin{theorem}
%There exists $C>0$ such that if $\theta^2 k>C$ then there is an inference function in $\mathbf{TC}^0$.
%\end{theorem}

%\begin{theorem}
%For all $\theta$ and $k$, there is a posterior function in $\mathbf{NC}^1$.
%\end{theorem}

%\begin{theorem}
%There exist $\theta$ and $k$ for which all posterior functions are $\mathbf{NC}^1$-hard to compute.
%\end{theorem}

%\fbox{{\em To be continued.}}

\subsection{Circuit Classes}

Here we give the formal definitions for the circuit classes that we will be interested in:

\begin{definition}
The circuit class $\mathbf{AC}^0$ is the class of constant depth circuits with a polynomial number of AND, OR and NOT gates, where the AND and OR gates have unbounded fan-in.
\end{definition}

It is well-known that there are explicit functions (such as the parity function) for which we can prove lower bounds against $\mathbf{AC}^0$ \cite{furst1984parity}.

\begin{definition}
The circuit class $\mathbf{NC}^1$ is the class of logarithmic depth circuits with a polynomial number of AND, OR and NOT gates, where the AND and OR gates have fan-in two. 
\end{definition}

In the broadcast tree model, the depth of the tree is logarithmic in the number of leaves. It follows that the posterior distribution on the root can always be computed in $\mathbf{NC}^1$.  %The usual approach is through {\em random restrictions} where for every input $x_i$ we leave it unset with probability $p$ and otherwise we set it to zero with probability $\frac{1-p}{2}$ and set it to one with the remaining probability $\frac{1-p}{2}$. The main insight is that if the parameters are chosen appropriately, with high probability the $\mathbf{AC}^0$ circuit becomes much simpler (while the parity function remains a parity on fewer inputs). 

\begin{definition}
A linear threshold function $f: \{0, 1\}^m \rightarrow \{0, 1\}$ takes the form $f(x) = \mbox{sgn}(w^T x - \theta)$ where $w \in \mathbb{R}^m$ and $\theta \in \mathbb{R}$. The circuit class $\mathbf{TC}^0$ is the class of constant depth circuits with a polynomial number of linear threshold function gates with unbounded fan-in. 
\end{definition}

The class $\mathbf{TC}^0$ is contained in $\mathbf{NC}^1$ and can compute any symmetric function of its inputs. In many ways, $\mathbf{TC}^0$ represents the frontier in circuit complexity. Impagliazzo, Paturi and Saks \cite{impagliazzo1997size} showed that depth $d$ $\mathbf{TC}^0$ circuits  with $m$ inputs need at least $m^{1+c^{-d}}$ wires to compute the parity function for some constant $c > 0$. Chen and Tell \cite{chen2019bootstrapping} showed that bootstrapping $\mathbf{TC}^0$ lower bounds just beyond this would yield super-polynomial lower bounds. Miles and Viola \cite{miles2015substitution} gave a candidate pseudorandom function computable in $\mathbf{TC}^0$ which helps explain the difficulty in proving lower bounds against $\mathbf{TC}^0$. 

\section{Lower bounds against $\mathbf{AC}^0$ for detection}

We show that there is no $\mathbf{AC}^0$ circuit that solves the detection problem for any non-trivial choice of parameters. In order to prove this, we are going to define a series of random projections that preserve the probability distribution of $X^{(d)}$ but reduce any circuit in $\mathbf{AC}^0$ to a constant with high probability. For the most part, the proof that these projections reduce the circuit to a constant will be a fairly standard argument using the switching lemma \cite{furst1984parity, yao1985separating, hastad1986almost}. However, due to the nature of the $X^{(d')}$, each projection will only fix a constant fraction of the variables, which will force us to apply $\Theta(\log n)$ successive projections every time we wish to reduce the circuit depth by one. The key observation is that we can preserve the probability distribution of $X^{(d)}$ by setting each vertex's label to its parent's label with probability $\theta$ and a random value otherwise. We prove:

\begin{theorem}\label{thm:ac0main}
Let $f:\{0,1\}^{L_d}\rightarrow \{0,1\}$ be computed by an $\mathbf{AC}^0$ circuit. Then there exists $\delta>0$ such that 
$\mathbb{P}[f(X^{(d)})=X^{(0)}]=1/2+O(n^{-\delta})$
\end{theorem}

\noindent We defer the proof to Appendix~\ref{app:ac0}. As usual, the key idea is to prove that $f$ can be approximated by a small DNF, although here the input to $f$ comes from the broadcast tree model.

\section{$\mathbf{NC}^1$-completeness of posterior functions}

%Rather than continuing to analyse what the requirements to compute the root community with nontrivial accuracy are, we will consider a harder variant of the problem, namely the following.

%\begin{definition}
%Let $f:\{0,1\}^n\rightarrow\{0,1\}$ be a function. $f$ solves approximate reconstruction if for every $x_1,\cdots,x_n$ such that $\mathbb{P}[X^{(0)}=1|X^{(d)}=x]\not\in (1/3,2/3)$, it is the case that $f(x_1,\cdots,x_n)$ is the community that $X^{(0)}$ is more likely to be in if  $X^{(d)}=x$.
%\end{definition}

In this section we will prove that 
\begin{theorem}\label{thm:nc1main}
For all $\theta$ and $k$ and in the Ising tree model, there are posterior functions in $\mathbf{NC}^1$. 
Moreover there are $\theta$ and $k$ for which posterior functions Ising model is $\mathbf{NC}^1$-hard problem.  

 \end{theorem}

%{\color{blue} In other words, we want an algorithm that determines what community the root is most likely to be in whenever there is a community that the root is at least $2/3$  likely to be in. We claim that this problem is $\mathbf{NC}^1$-complete. 
We begin by proving the first part of the theorem \---- i.e. that computing the posterior can be solved in $\mathbf{NC}^1$. The obvious approach to establish this would be to try to compute the probability distribution of each node's label based on the probability distributions of its children's labels. However, this could fail due to rounding errors. Instead, we will show that for each node, there exists a random $\mathbf{NC}^1$ function that sometimes outputs a label for that node, such that the probability distribution of the label it output given that it outputs one is the same as the probability distribution of the vertex's label. A little more precisely, for each $d'$ we will show that there exists a random function $F:\{0,1\}^{L_{d'}}\to \{0,1,?\}$ that can be computed by an NC circuit of depth $O(d')$ such that for every $x\in\{0,1\}^{L_{d'}}$, the probability that $F(x)= '?'$ is reasonably small and
\[
\mathbb{P}[F(x)=1|F(x)\ne'?'] = \mathbb{P}[X^{(0)}=1|X^{(d')}=x]
\]

In order to prove this, we will induct on $d'$. If we assume that we have such a probability distribution for $d'-1$, then we can use it to guess the value of $X^{(1)}$ based on the value of $X^{(d')}$. Then, we can also guess which children of $v_1^{(0)}$ have the same label as it, and if any choice of $X^{(0)}$ is consistent with all of these guesses, we can conclude that $X^{(0)}$ has that value. This would give a suitable probability distribution for $d'$, except that it has an excessively high probability of returning $'?'$. Fortunately, we can deal with that by trying it multiple times and returning the first value in $\{0,1\}$ that we get. We defer the proof to Appendix~\ref{app:nc1part1}.

%\note{I am not sure what "computing the posterior means." Is it defined as "what a posterior function does"? Does it mean giving an estimate of $\mathbb{P}[X^{(0)}=1|X^{(d)}=x]$ that is always within $o(1)$ of the true value? In the later case, we need to show that the error is always in $o(1)$, not just that it can always distinguish cases in which it is at least $1/2+\delta_d$ from cases where it is at most $1/2-\delta_d$.}

%\begin{remark}
%If we wanted a more precise output than simply the most likely community of the root, we could compare exactly how many of the $F_i$ returned $1$ and $0$ in order to estimate the probability that the root is in community $1$. For any $c>0$, it would be possible to ensure that our estimate was within $O(n^{-c})$ of the true probability for all possible inputs, provided we used $n^{2c+2}$ functions $F_i$ instead of $n^2$.
%\end{remark}

For the second part of the theorem we interpret any node that is very likely to have a label of $1$ to be a variable that is actually set to $1$, and similarly for the label $0$.  Then, we will construct gadgets for AND and OR, at which point it will be easy to translate an arbitrary $\mathbf{NC}^1$ circuit to an $\mathbf{NC}^0$ formula for $X^{(d)}$ in terms of the circuit's inputs. We defer the proof to Appendix~\ref{app:nc1part2}

%\fbox{Wait, is $\mathbf{NC}^1$ supposed to be uniform? I may have to amend something.}

%\fbox{Also, I should consider formally defining what it means to track a function.}

\section{A $\mathbf{TC}^0$ circuit for inference}
The previous result implies that if $\mathbf{TC}^0\ne \mathbf{NC}^1$ then no $\mathbf{TC}^0$ circuit can compute a posterior function in the Ising tree model. However we can still hope that $\mathbf{TC}^0$ circuits attempting to determine $X^{(0)}$ can still perform well in the average case and can compute an inference function.

A natural approach is to guess that the root has the same label as the majority of the leaves, which gives the right answer with probability $1+\Omega(1)$ if $\theta>1/\sqrt{k}$. However, this is not an inference function. In particular, it achieves worse error even in an average-case sense. Alternatively we could compute an inference function using belief propagation but the naive way to encode this as a circuit would lead to logarithmic depth. The key idea is that the function computed by belief propagation is robust to injecting noise at the leaves. We use this idea by first guessing that each node at depth $\lfloor\log_k(\log_2(n))\rfloor$ has the same label as the majority of the leaves descended from it. Then we guess the value of $X^{(0)}$ by computing the output of belief propagation (on the smaller depth tree) using a look up table. We are able to prove that this circuit is indeed a posterior function when $k\theta^2$ is sufficiently large and we conjecture that it is for any $k \theta^2 > 1$.

More precisely we will build a $\mathbf{TC}^0$ circuit that encodes the following algorithm.

\begin{algorithm}
{\sc LinearizedBP}(d, k, $\theta$, $X^{(d)}$, $f$)
\begin{enumerate}
\item Let $d'=\lfloor\log_k(\log_2(n))\rfloor$.

\item For each $i\in L_{d'}$, randomly select $x^\star_i\in\{0,1\}$ and set 

\[x_i=
\begin{cases}
1 &\text{ if } \sum_{j\in L_{d-d'}(i)} X^{(d)}_j> k^{d-d'}/2\\
0 &\text{ if } \sum_{j\in L_{d-d'}(i)} X^{(d)}_j< k^{d-d'}/2\\
x^\star_i &\text{ if } \sum_{j\in L_{d-d'}(i)} X^{(d)}_j= k^{d-d'}/2\\
\end{cases}
\]

\item Output $f(x)$
\end{enumerate} 
\end{algorithm}
First of all, note that each value of $n$ has a unique corresponding value of $d'$, and each of the $x_i$ can be computed from the inputs and a random bit by a threshold gate. $k^{d'}\le\log_2(n)$, so there are at most $n$ possible values of $x$. That means that we can use an AND gate to check for each possible value of $x$ and then OR together the ones for which $f(x)=1$. That means that for any fixed series of functions $f_d:\{0,1\}^{k^{\lfloor\log_k(\ln(n))\rfloor}}\to\{0,1\}$, there is a $\mathbf{TC}^0$ circuit that computes {\sc LinearizedBP}(d, k, $\theta$, $X^{(d)}$, $f$) given access to $\log_2(n)$ random bits. Furthermore, we conjecture the following.

\begin{conjecture}
There exists a series of functions $f_d:\{0,1\}^{k^{\lfloor\log_k(\ln(n))\rfloor}}\to\{0,1\}$ such that if $X'={\sc LinearizedBP}(d, k, \theta, X^{(d)}, f_d)$ then
\[\lim_{n\to\infty} \mathbb{P}[X'=X^{(0)}]- \mathbb{P}[\BP(X^{(d)})=X^{(0)}]=0,
%\sum_{x\in \{0,1\}^n} \max\left(\mathbb{P}[X^{(0)}=0,X^{(d)}=x],\mathbb{P}[X^{(0)}=1,X^{(d)}=x]\right)=0\]
\]
where $\BP(x) : \{0,1\}^{L_d} \to \{0,1\}$ returns the more likely posterior label of the root
\[
BP(x) = a \; \mbox{ if } \; \mathbb{P}[X^{(0)} = a | X^{(d)} = x] > \mathbb{P}[X^{(0)} = 1- a | X^{(d)} = x] 
\]
\end{conjecture}

In other words, we believe that {\sc LinearizedBP} can compute $X^{(0)}$ with optimal accuracy. If $k\theta^2\le 1$, then it is known that no algorithm can compute $X^{(0)}$ from $X^{(d)}$ with nontrivial accuracy, so this algorithm uninterestingly attains optimal accuracy. In this section, we will prove that there exists $C>1$ such that {\sc LinearizedBP} can attain optimal accuracy whenever $k\theta^2>C$. The case where $1<k\theta^2\le C$ remains open. The first step towards proving that it can attain optimal accuracy for large values of $k\theta^2$ is to prove that when the algorithm is run, $x$ is a reasonably accurate approximation of $X^{(d')}$. For that, we need the following standard second moment lemma which we include for completeness in Appendix~\ref{app:deviation} (similar lemmas were proven in previous work including~\cite{EvKePeSc:00}). 
%{\em El: The lemma proof should definitely go into the appendix}

\begin{lemma}\label{lem:deviation}
For any $d$, $k$, and $\theta$ such that $k\theta^2>2$, 
\[\mathbb{P}\left[\sum_{i=1}^{k^d} X^{(d)}_i\le k^d/2\middle | X^{(0)}=1\right]\le\frac{1}{\theta^2k-1}\]
\end{lemma}

By symmetry, this also implies that $\mathbb{P}\left[\sum_{i=1}^{k^d} X^{(d)}_i\ge k^d/2\middle | X^{(0)}=0\right]\le\frac{1}{\theta^2k-1}$. So, that gives us a bound on $\mathbb{P}[x_i\ne X^{(d')}_i]$ when the algorithm is run. That leaves the task of showing that we can determine $X^{(0)}$ with optimal accuracy from a noisy version of $X^{(d')}$. In order to discuss the accuracy with which one can do that, we will need to define the following.

\begin{definition}
Let $0\le s\le 1/2$ and $d$ be a positive integer. Also, let $X'\in\{0,1\}^{L_d}$ such that for each $i$, $X'_i$ is independently set equal to $1-X^{(d)}_i$ with probability $s$ and $X^{(d)}_i$ otherwise.
\[P_{s,d}=\sum_{x\in \{0,1\}^{L_d}} \max(\mathbb{P}[X^{(0)}=0,X'=x],\mathbb{P}[X^{(0)}=1,X'=x])\]
\end{definition}
In other words, $P_{s,d}$ is the maximum accuracy with which we can determine $X^{(0)}$ from a noisy version of $X^{(d)}$ in which each bit is flipped with probability $s$.  Mossel et al. \cite{MoNeSl:16b} show the following:

\begin{proposition}\label{noisybp}\cite{MoNeSl:16b}
There exists $C>0$ such that if $k\theta^2>C$ then
\[\lim_{s\to 1/2}\inf \lim_{d\to\infty}\inf P_{s,d}=\lim_{d\to\infty}\inf P_{0,d}\]
\end{proposition}
 In other words, if $k\theta^2$ is sufficiently large then the maximum accuracy with which $X^{(0)}$ can be determined from a highly noisy estimate of $X^{(d')}$ is the same as the maximum accuracy with which $X^{(0)}$ can be determined from $X^{(d')}$. That allows us to prove that {\sc LinearizedBP} is optimal for large values of $k\theta^2$. More formally, we have the following:
\begin{theorem}\label{thm:tc0main}
There exists $C'>0$ such that if $k\theta^2>C'$ then there exists a function $f$ for which {\sc LinearizedBP} run on $f$ is an inference function for the Ising model on trees. 
\end{theorem}

\begin{proof}
First, let $C'=\max(C,4)$.  First we observe that for any $d$, when {\sc LinearizedBP} is run, each bit $x_i$ is independently set equal to $X^{(d')}_i$ with some advantage over random guessing and set to the opposite value otherwise. Let $s_d=\mathbb{P}[x_i\ne X^{(d')}_i]$. Next, let $f_d$ be the function that maximizes the probability that {\sc LinearizedBP} outputs the correct label for the root. Let $q$ be the probability that it succeeds. Then we have
$$ q =\sum_{x'\in \{0,1\}^{L_{d'}}} \max(\mathbb{P}[X^{(0)}=0,x=x'],\mathbb{P}[X^{(0)}=1,x=x'])=P_{s_d,d'}$$
Now, let $s'=\frac{1}{\theta^2k-1}$. Proposition~\ref{noisybp} shows that $s_d\le s'$ for all $d$, and adding more noise can never make it easier to determine $X^{(0)}$, so for every $d$, it must be the case that
\[P_{0,d'}\ge P_{s_d,d'}\ge P_{s',d'}\]
Combining that with the previous theorem shows that
\[\lim_{d'\to\infty}\inf P_{0,d'}\ge \lim_{d'\to\infty}\inf P_{s',d'}\ge \lim_{s\to 1/2}\inf \lim_{d'\to\infty}\inf P_{s,d'}=\lim_{d'\to\infty}\inf P_{0,d'}\]
Also, $P_{0,d'}$ is a nonincreasing function of $d'$, so $P_{0,d'}$ converges. So, 
\[\lim_{s\to 1/2}\sup \lim_{d'\to\infty}\sup P_{s,d'}\le \lim_{d'\to\infty}\sup P_{s',d'}\le \lim_{d'\to\infty}\sup P_{0,d'}=\lim_{d'\to\infty} P_{0,d'}\]
That implies that all of these sequences converge to $\lim_{d'\to\infty} P_{0,d'}$, and thus that {\sc LinearizedBP} computes $X^{(0)}$ with optimal accuracy.
\end{proof}

%\fbox{{\em This would still attain optimal accuracy if I always have $x_i=1$ in the event of a tie}}

%\fbox{{\em However, that would be harder to prove.}}

\section{$\mathbf{NC}^1$ hardness of detection with many labels}

So far, we have been assuming that there are only two labels that could be assigned to a vertex. However, we could instead have $m$ labels for arbitrary $m$. That leads to the following definition

\begin{definition}[Generalized broadcast process on a tree]
Given parameters $m>0$ and an $m\times m$ matrix $M$ with nonnegative entries and columns that add up to $1$, the {\em generalized broadcast process on  $T$} is
 an $m$-state Markov process $\{\sigma^\star_u : u \in T\}$ defined as follows:
let $\sigma^\star_\rho$ be drawn uniformly at random from $\{1,\cdots,m\}$. Then, for each $u$ such that $\sigma^\star_u$ is defined,
independently for every $v \in L_1(u)$ let $\sigma^\star_v = 1$ with probability $M_{i,\sigma^\star_u}$ for each $i$.
\end{definition}

In other words, in the generalized broadcast model, the root is randomly assigned a label in $\{1,\cdots,m\}$, and then each other vertex is assigned a label with a probability distribution corresponding to the column of $M$ indexed by its parent's label. 
 Note that the previous case is simply the instance of this where $m=2$ and $M=\theta I+\frac{1-\theta}{2} J$, where $J$ is the matrix with all entries equal to $1$. There is an important difference between the case when there are just two labels and when there are more. It turns out there are many natural cases where it is possible to detect the label of the root, but not by taking the majority vote of the labels of the leaves. The function computed by Belief Propagation is generally more complicated and the main result of this section is to show that this manifests as a phase transition in the circuit complexity of solving detection. When there are many labels, we will show that the problem becomes $\mathbf{NC}^1$ hard. 
 
% For the most part, our results about the previous case carry over. However, one notable difference is that there are values of $m$ and $M$ such that there are detection functions in $\mathbf{NC}^1$ but not in $\mathbf{TC}^0$ unless $\mathbf{TC}^0=\mathbf{NC}^1$.

First, we need a problem that is $\mathbf{NC}^1$-hard in the average case. In a celebrated result, Barrington showed that deciding whether the word problem (i.e. if a given word is the identity or not) over a finite nonsolvable group is $\mathbf{NC}^1$-complete \cite{barrington1989bounded}. We will work with the alternating group $A_5$:

%In order to prove that $\mathbf{TC}^0$ circuits fail at reconstructing the root's community with nontrivial accuracy if $\mathbf{TC}^0\ne \mathbf{NC}^1$, we need a problem that is $\mathbf{NC}^1$-hard in the average case. For this, we use the problem of computing the product of a series of elements of the alternating group $A_5$. This problem is NC-1 complete in the following sense.

\begin{proposition}\label{prop:barr} \cite{barrington1989bounded}
For every $c\in A_5$ such that $c\ne 1$, determining whether a product of elements of $A_5$, $\prod_{i=1}^n \sigma_i$ is $c$ or the identity given that it is one of them is $\mathbf{NC}^1$-complete.
\end{proposition}

Conveniently, this problem has a simple worst-case to average-case reduction:

\begin{theorem}
Let $f_r: A_5^r\rightarrow A_5$ be a family of functions. Suppose there exists $\epsilon>0$ independent of $r$ such that when $\Sigma_1,\cdots,\Sigma_r$ are independently drawn from $A_5$ according to the uniform distribution, $$\mathbb{P}[f_r(\Sigma_1,\cdots,\Sigma_r)=\prod_{i=1}^r \Sigma_i]\ge 1/60+\epsilon$$ If $\mathbf{TC}^0\neq \mathbf{NC}^1$ then there is no $\mathbf{TC}^0$ circuit that computes $f$.
\end{theorem}

\begin{proof}
For the sake of contradiction, we will assume that there is a $\mathbf{TC}^0$ circuit that computes $f$. Let $h_n:\{0,1\}^n\rightarrow \{0,1\}$ be an $\mathbf{NC}^1$-complete family of functions. Consider the following randomized algorithm attempting to compute $h_n(x)$. First, generate a random $c\in A_5\backslash \{1\}$. 
Next, the completeness of $h_n$ implies there there exists $r$ polynomial in $n$ and $\sigma\in A_5^r$ such that $\prod_{i=1}^r \sigma_i=c$ if $h_n(x)=1$ and $\prod_{i=1}^r \sigma_i=1$ if $h_n(x)=0$ (note that $\sigma$ depends on $c$ and $x$ and the computation of $\sigma$ is in $\mathbf{NC}^0$). 
%Consider the functions $g_n : A_5^n \to \{0,1\}$ that return the value $1$ if $\prod_{i=1}^r \sigma_i=c$ and return the value $0$ if $\prod_{i=1}^r \sigma_i=1$. 
%From Proposition~\ref{prop:barr} these collection of functions is complete. 
Now 
randomly select $b_i\in A_5$ for each $1\le i\le r$. Next compute 
$$f(\sigma_1 b_1, b_1^{-1} \sigma_2 b_2, b_2^{-1} \sigma_3 b_3,\cdots, b_{r-1}^{-1} \sigma_r b_r).$$ 
If it is equal to $b_r$, conclude that $h_n(x)=0$, if it is $c b_r$ then conclude that $h_n(x)=1$, and output nothing otherwise. No matter what the value of $\sigma$ is, the probability distribution of $(\sigma_1 b_1, b_1^{-1} \sigma_2 b_2,\cdots, b_{r-1}^{-1} \sigma_r b_r)$ is the uniform distribution on $A_5^r$. Hence we have that $$\mathbb{P}[f(\sigma_1 b_1,\cdots,b_{r-1}^{-1} \sigma_r b_r)=\sigma_1\sigma_2,\cdots,\sigma_r b_r]\ge 1/60+\epsilon$$ Thus, this algorithm computes $h_n(x)$ correctly with a probability of at least $1/60+\epsilon$. Furthermore, $c$ is independent of  $(\sigma_1 b_1,\cdots, b_{r-1}^{-1} \sigma_r b_r)$, and thus of what $f$ will return if it computes the product incorrectly. So, this algorithm computes $h_n(x)$ incorrectly with a probability of at most $1/60$. Thus if we repeat this process a large polynomial number of times and take the majority vote, we can compute $h_n(x)$ correctly with probability at least $1-o(2^{-n})$. Thus there must be some choices of our random variables for which this computes $h_n(x)$ correctly for every $x$. This whole procedure can be carried out by a $\mathbf{TC}^0$ circuit, so $\mathbf{TC}^0$=$\mathbf{NC}^1$.
\end{proof}

Now that we have a problem that is $\mathbf{NC}^1$-hard in the average case, we need a way to reduce this to the problem of determining the label of the root for some choice of parameters. In order to do that, we consider the following instance of the generalized broadcast process on a tree. There is one label for every ordered pair $(\sigma, \sigma')\in A_5^2$, and $k=60000$. Given a vertex with a parent with label $(\sigma,\sigma')$, we select a random $b\in A_5$. Then, we set its label to $(b,b^{-1}\sigma)$ with probability $2/3$ and $(b,b^{-1}\sigma')$ with probability $1/3$. In other words, each child of a vertex is assigned a random ordered pair that multiplies to $\sigma$ with probability $2/3$ and a random ordered pair that multiplies to $\sigma'$ with probability $1/3$. For the rest of this section, we will assume that $\sigma^\star$ was generated by the generalized broadcast process with these parameters.

Note that it is straightforward to implement this process with an $\mathbf{NC}^1$ circuit because the tree has logarithmic depth.
Moreover, we argue that detection is information-theoretically possible. The key idea is for any $d'$, if we can determine the labels of the vertices at depth $d'$ so that each label is correct (independently) with probability $0.99$ then for any vertex at depth $d'-1$ we can determine its label with probability at least $0.99$. We do this by taking the two most common products of the elements among its children's suspected labels and by a Chernoff bound it is easy to see that this procedure succeeds with probabiity at least $0.99$. Furthermore because the subtrees of each vertex at depth $d'-1$ are disjoint, the probability our guess is correct is independent. Now we can continue this process until we reach the root. This type of recursive reconstruction arguments are by now standard, see e.g. {\em Mossel, Mossel-Peres}

Next we will give an alternative procedure for sampling from the generalized broadcast tree model. This will allow us to embed the word problem for $A_5$ equivalently as the problem of guessing the label of the root. 

%Note that for any $d'$, if we can determine the labels of vertices at depth $d'$ with $99\%$ accuracy, then for any  vertex at depth $d'-1$, we can determine its label with at least $99\%$ accuracy by taking the two most common products of the elements in its children's suspected labels. As such, it is possible to recover the root's label from the leaves' communities with an accuracy of at least $99\%$. In order to encode a product of elements in $A_5$ as an instance of the root reconstruction problem, we will use the following algorithm to assign every vertex in the tree a label that is a function of variables $\sigma_1,\cdots,\sigma_{2^{d+1}}$.

\begin{algorithm}
productTreeConstructionAlgorithm(d):
\begin{enumerate}
\item Set $\overline{X}^{(0)}=(\sigma_1\cdot \sigma_2\cdot\cdots\cdot \sigma_{2^d},$  $\sigma_{2^d+1}\cdot \sigma_{2^d+2}\cdot\cdots\cdot \sigma_{2^{d+1}})$.

\item For $d' = 1$ to $d$
\begin{enumerate}
\item For each $i \in L_{d'}$:
\begin{enumerate}
\item There will exist a constant $1\le j\le 2^{d+1}$ and constants $b,b',b''\in A_5$ such that $$\overline{X}^{(d'-1)}_{\parent(i)}=(b'\cdot \sigma_j\cdot\cdots\cdot \sigma_{j+2^{d-d'+1}-1}\cdot b,\mbox{ }b^{-1}\cdot \sigma_{j+2^{d-d'+1}}\cdot\cdots\cdot \sigma_{j+2^{d-d'+2}-1}\cdot b'')$$

\item Randomly select $b'''\in A_5$.

\item With probability $2/3$, set $$\overline{X}^{(d')}_i=(b'\cdot \sigma_j\cdot\cdots\cdot \sigma_{j+2^{d-d'}-1}\cdot b''',\mbox{ }(b''')^{-1}\cdot \sigma_{j+2^{d-d'}}\cdot\cdots\cdot \sigma_{j+2^{d-d'+1}-1}\cdot b)$$ Otherwise, set $$\overline{X}^{(d')}_i=(b^{-1}\cdot \sigma_{j+2^{d-d'+1}}\cdot\cdots\cdot \sigma_{j+3\cdot 2^{d-d'}-1}\cdot b''',\mbox{ }(b''')^{-1}\cdot \sigma_{j+3\cdot 2^{d-d'}}\cdot\cdots\cdot \sigma_{\cdot 2^{d-d'+2}-1}\cdot b'')$$
\end{enumerate}
\end{enumerate}

\item Return $\overline{X}^{(d)}$.
\end{enumerate}
\end{algorithm}

In step 2.a.i we asserted that every element of $\overline{X}^{(d')}$ will have the form
$$(b'\cdot \sigma_j\cdot\cdots\cdot \sigma_{j+2^{d-d'}-1}\cdot b,\mbox{ }b^{-1}\cdot \sigma_{j+2^{d-d'}}\cdot\cdots\cdot \sigma_{j+2^{d-d'+1}-1}\cdot b'')$$
It is easy to see that this is true for $\overline{X}^{(0)}$ and throughout the process, $\overline{X}^{(d'-1)}$ will always be set to an expression of this form. The key fact is:

\begin{lemma}\label{lem:gen}
Let $\sigma\in A_5^{2^{d+1}}$ and $x_0=\left(\prod_{i=1}^{2^d} \sigma_i,\prod_{i=2^d+1}^{2^{d+1}} \sigma_i\right)$. %Let $\overline{X}^{(d)}$ be the output of $productTreeConstructionAlgorithm(d)$. 
Then for every $x\in (A_5^2)^n$,
\[\mathbb{P}\left[X^{(d)}=x\middle|X^{(0)}=x_0\right]=\mathbb{P}\left[\overline{X}^{(d)}(\sigma)=x\right]\]
\end{lemma}

Thus $productTreeConstructionAlgorithm(d)$ is an equivalent way to sample from the generalized broadcast tree model that we defined earlier. 

\begin{proof}
We will prove by induction on $d'$ that the distribution of $X^{(d')}$ given $X^{(0)}=x_0$ is identical to the distribution of $\overline{X}^{(d')}(\sigma)$ for every $d'$. If $d'=0$, then $\overline{X}^{(d')}=x_0$, so the base case holds. Now, assume that it holds for $d'-1$. 

It is easy to check that the way we have defined step 2.a.iii every 
vertex at depth $d'$ is assigned a label whose product is equal to the first permutation in its parent's label with probability $2/3$ and the second permutation in its parent's label with probability $1/3$. Moreover the pair of permutations is chosen uniformly at random subject to this constraint. Finally each element of $\overline{X}^{(d')}$ is independent conditioned on the value of its parent. These are exactly the key properties that defined our generalized broadcast tree model, and hence completes the proof. %With probability $2/3$, it is assigned a value that multiplies to the first element of its parents value and with probability $1/3$ it is assigned a value that multiplies to the second element of its parent's value. 
%Furthermore, for any fixed value of $\sigma$, each element of $\overline{X}^{(d')}$ is equally likely to evaluate to any element of $A_5^2$ with the desired product, even conditioned on the full value of $X^{(d'-1)}$. So, the probability distribution of $\overline{X}^{(d')}(\sigma)$ is exactly the same as the probability distribution of $X^{(d')}$ given $X^{(0)}=x_0$. The claim follows by induction, and the $d'=d$ case gives us the desired conclusion.
\end{proof}

%So, any algorithm that can recover $X^{(0)}$ from $X^{(d)}$ can be used in conjunction with this algorithm to multiply lists of elements in $A_5$. In particular, we have the following.

Now we are ready to prove that any algorithm for solving the detection problem for our generalized broadcast tree model can be used to solve the word problem over $A_5$ with some advantage over random guessing:

\begin{theorem}
Let $g_d: (A_5^2)^{k^d}\rightarrow A_5$ be a family of functions. Suppose there exists $\epsilon > 0$ independent of $d$ such that  $$\mathbb{P}[g_d(X^{(d)})=X^{(0)}]\ge \frac{1}{|A_5|^2} +\epsilon$$ If $\mathbf{TC}^0\neq\mathbf{NC}^1$ then $g$ is not in $\mathbf{TC}^0$.
\end{theorem}

\begin{proof}
For the sake of contradiction we will assume that $g\in \mathbf{TC}^0$. Let $\Sigma_1,\cdots,\Sigma_{2^{d+1}}$ be chosen randomly. We can interpret $productTreeConstructionAlgorithm(d)$ as outputting a random formula that labels the leaves of the generalized broadcast tree model. The key point is both the depth of the tree and the number of bits of randomness that determine the value at any leaf are both logarithmic. Thus $X^{(d)}$ can be computed by a $\mathbf{TC}^0$ circuit. Now let $g'_d$ be the composition of $g_d$ and $productTreeConstructionAlgorithm(d)$. 

Because $g_d$ solves the detection problem we have that
\[\mathbb{P}\left[g'_d(\Sigma_1,\cdots,\Sigma_{2^{d+1}})=\left(\prod_{i=1}^{2^d} \Sigma_i,\prod_{i=2^d+1}^{2^{d+1}} \Sigma_i\right)\right]\ge \frac{1}{|A_5|^2}+\epsilon\]
where the randomness is over both the choice of the $\Sigma_i$'s and $g'$ which depends on the generation process. For the sake of simplifying the notation, let $g'_d(\sigma)=(g^{[1]}_d(\sigma),g^{[2]}_d(\sigma))$. Now there are two cases:

In the first case suppose that $g^{[1]}_d$ gets nontrivial advantage over random guessing. In particular suppose
$$\mathbb{P}\left[g^{[1]}_d(\Sigma_1,\cdots,\Sigma_{2^{d+1}})=\prod_{i=1}^{2^d} \Sigma_i\right]\ge\sqrt{\frac{1}{|A_5|^2}+\epsilon}$$
There must exist a specific choice $\Sigma_{2^d + 1} = \sigma_{2^d +1}, \cdots, \Sigma_{2^{d+1}} = \sigma_{2^{d+1}}$ and setting of the randomness in the generation process which achieves nontrivial advantage over random guessing. Even when we fix these values, the function is still in $\mathbf{TC}^0$ and hence we conclude $\mathbf{TC}^0=\mathbf{NC}^1$.

In the second case, we must have
$$\mathbb{P}\left[g^{[2]}_d(\Sigma_1,\cdots,\Sigma_{2^{d+1}})=\prod_{i=2^d+1}^{2^{d+1}} \Sigma_i\middle|g^{[1]}_d(\Sigma_1,\cdots,\Sigma_{2^{d+1}})=\prod_{i=1}^{2^d} \Sigma_i\right]\ge\sqrt{\frac{1}{|A_5|^2}+\epsilon}$$
The idea is we want to use $g^{[2]}_d$ to solve an $\mathbf{NC}^1$ hard problem, but to do so using the above inequality we need to decide if the output of $g^{[1]}_d$ is correct. Now we can once again use an average-case reduction to reduce to the case when we know the product of the inputs to $g^{[1]}_d$ and thus check its own output. 

In particular for any $\sigma_1,\cdots,\sigma_{2^{d+1}}\in A_5$ and randomly generated $B_1,\cdots,B_{2^{d+1}}$, let
$$\Sigma'=(\sigma_1 B_1,B_1^{-1}\sigma_2 B_2, B_2^{-1} \sigma_3 B_3,\cdots,B_{2^{d+1}-1}^{-1}\sigma_{2^{d+1}}B_{2^{d+1}})$$
The distribution of $\Sigma'$ is uniform on $A_5^{2^{d+1}}$ so we have
$$\mathbb{P}\left[g^{[2]}_d(\Sigma')=B_{2^d}^{-1}\left(\prod_{i=2^d+1}^{2^{d+1}} \sigma_i\right)B_{2^{d+1}}\middle|g^{[1]}_d(\Sigma')=\left(\prod_{i=1}^{2^d} \sigma_i\right)B_{2^d}\right]\ge\sqrt{\frac{1}{|A_5|^2}+\epsilon}$$
Now we can choose $\sigma_1,\cdots,\sigma_{2^d}$ such that we already know their product and we can repeatedly generate $B_1,\cdots,B_{2^{d+1}}$ until we find one for which $$g^{[1]}_d(\Sigma')=\left(\prod_{i=1}^{2^d} \sigma_i\right)B_{2^d}$$
Now if we guess that $\prod_{i=2^d+1}^{2^{d+1}}\sigma_i$ is equal to $B_{2^d} g^{[2]}_d(\Sigma') B_{2^{d+1}}^{-1}$ we will get nontrivial advantage over random guessing. As before there must be some choice of the randomness (in this case the values of $B_1, \cdots, B_{2^{d+1}}$ and the randomness in the generation process) where the probability of computing the product is at least average. This again implies that $\mathbf{TC}^0=\mathbf{NC}^1$.
\end{proof}

So, this is a set of parameters for which one can determine the root's label from the leaves' labels with very high accuracy in the average case. However, unless $\mathbf{TC}^0$=$\mathbf{NC}^1$, there is no $\mathbf{TC}^0$ algorithm that can determine the root's label with an accuracy that is nontrivially higher than that attained by guessing blindly. With some more work, we could prove that this also holds for sufficiently slight perturbations of these parameters. In Appendix~\ref{app:labelred} we show how to reduce the number of labels to $16$ by using symmetry arguments and working with conjugacy classes of permutations instead. 

\section{Difficulty of generation} \label{sec:generation}

In this paper, we are mostly concerned with depth lower (and upper) bounds for estimating $X^{(0)}$ given $X^{(d)}$. However, we also study the generation problem, i.e., the complexity of generating $X^{(d)}$ given a sequence of random bits as an input. More formally: 

\begin{definition}
We say that a series of functions $f_d:\{0,1\}^{m(d)} \to \{0,1\}^{L_d}$ are {\em generation} functions if under the uniform distribution over the inputs, it holds that $f_d(x)$ has the distribution $X^{(d)}$ for all $d$.
%if $(X^{(0)},\cdots,X^{(d)})\sim TBM_d(m,k,M)$ then
%\[
%\mathbb{P}[f(X^{(d)})=X^{(0)}]\ge 1/2+\delta
%\]
We call such functions $(\delta_d)_{d=1}^{\infty}$-approximate-generation functions if the total variation distance between the distribution of $f_d(x)$ and $X^{(d)}$ is bounded by $\delta_d$ for all $d$
\end{definition}

Despite the fact that the tree has logarithmic depth, it turns out that generation can be accomplished in $\mathbf{AC}^0$ easily.

%Before we discuss the computational complexity of determining $X^{(0)}$ from $X^{(d)}$, we will consider the complexity of generating a string of bits with the appropriate probability distribution for $X^{(d)}$ from $X^{(0)}$ and some independent, random bits. {\color{blue} For any given $i$, $X^{(d)}_i$ is completely determined by $X^{(0)}$ and the behavior of the $d$ edges leading from $v^{(0)}_1$ to $v^{(d)}_i$. As such, we can brute force the task of determining $X^{(0)}_i$ from those pieces of information with an $\mathbf{AC}^0$ circuit.} More formally, we have the following.

\begin{theorem}\label{thm:genmain}
If $\theta$ is a dyadic number: $\theta = a/2^b$ for some integers $a$ and $b$, then there exists generation functions in $\mathbf{AC}^0$.
Moreover, for all $\theta$, and any constant $c > 0$, there exists $2^{-n^c}$--approximate-generation functions in $\mathbf{AC}^0$. 
\end{theorem} 

%Let $W_1,\cdots,W_{(kn-k)/(k-1)}$ be independent random variables such that for each $i$, $W_i$ is $1$ with probability $\epsilon$ and $0$ otherwise. %There exists a function $f$ that can be computed by an $\mathbf{AC}^0$ circuit such that if $X'=f(X^{(0)},W_1,W_2,\cdots,W_{(kn-k)/(k-1)})$ then for all $x^0\in %\{0,1\}$ and $x\in\{0,1\}^n$,
%\[\mathbb{P}[X^{(0)}=x^0,X^{(d)}=x]=\mathbb{P}[X^{(0)}=x^0,X'=x]\]
%\end{theorem}

\begin{proof}
Assume first that $\theta$ and therefore $(\theta \pm 1)/2$ are dyadic. 
This means that there exists a function 
$g : \{0,1\}^s \to \{0,1\}$ of a bounded number of bits such that $\mathbb{P}[g = 1] = (\theta+1)/2$. 
We apply a copy of $g$ independently for each vertex of the tree thus obtaining a collection of independent random variables $(Y_v)$. 
%First, associate $v_i^{(d')}$ with $W_{i+k^{d'-1}+k^{d'-2}+\cdots+k}$ for all $0<d'\le d$ and $1\le i\le k^{d'}$. This associates each non root vertex with a different one of the $W_i$, and $\mathbb{P}[X_i^{(d')}\ne X_{\lceil i/k\rceil}^{(d'-1)}]=\mathbb{P}[W_{i+k^{d'-1}+\cdots+k}=1]$. 
So, if we set $Y^{\prime (0)}$ to be a uniformly random bit and then define $X^{\prime}_v = \prod_{w \in \path(\rho,v)} Y_w$,
 then the probability distribution of $X^\prime$ is identical to the probability distribution of $X$. Furthermore, for each $v$, there are at most $d$ elements of $Y$ that effect the value of $X^{\prime}_v$. That means that there are only $2^{d+1}\le 2n$ possible values of $X^{(0)}$ and the elements of $Y$ that effect $X^{\prime(d)}_i$. As such, we only need $O(n)$ gates to have an AND for every possible combination of values of $X^{(0)}$ and these $Y_v$, at which point we can OR together all of the ones for values that result in $X^{\prime}_v = 1$. Doing this for every $v$ merely multiplies the number of gates by $n$, and this clearly has constant depth. This proves the first part of the theorem. 
 
 The second part of the theorem is similar, except we now approximate coin tosses of bias $(\theta+1)/2$. It is easy to see that an approximation to error $2^{-n^c}$ is achievable in $\mathbf{AC}^0$ in constant depth and size polynomial in $n$. 
 This is done by generating a polynomial number of unbiased bits $Z_1,\ldots,Z_{n^c+\lceil \log_2(2n)\rceil}$  and considering them as the binary expansion of a number in $[0,1]$. We then declare the bias-coin toss to be $1$ if the resulting number is bigger than $(1+\theta)/2$ and $0$ otherwise. The threshold computation 
 $\sum Z_i > (1+\theta)/2$ can be carried out by an OR of AND gates. 
 %Thus, we can define $f$ such that $f(X^{(0)},W_1,\cdots,W_{(nk-k)/(k-1)})=X^{\prime (d)}$ in order to satisfy the conditions of this theorem.
\end{proof}

%\fbox{\parbox{16cm}{I added a $\lceil \log_2(2n)\rceil$ term to the number of coins used to simulate each biased coin in order to ensure that all of them put together were off by at most $2^{-n^c}$. Incidentally, is there a term for functions that can be computed by a circuit with constant fanin and depth $O(\log\log n)$, because this computation can be performed by one (except for simulating the biased coin flips less accurately).}}

\begin{remark}
If we consider a computational model where the inputs have bias $\theta$ instead of $1/2$, then the proof above provides generation functions in $\mathbf{AC}^0$.
\end{remark}

%\begin{remark}
%If we require that each of the $W_i$ is $1$ with probability $1/2$, then that forces $\mathbb{P}[X^{(0)}=x^0,X'=x]$ to be an integer multiple of a power of $2$. For most values of $\epsilon$, that would ensure that the probability distribution of $f(X^{(0)},W_1,W_2,\cdots,W_m)$ will not exactly match the %probability distribution of $X^{(m)}$. However, it is possible to determine which of two numbers is larger with an $\mathbf{AC}^0$ circuit. Also, %$\mathbb{P}[W_{i}/2+W_{i+1}/4+\cdotsW_{i+m'-1}/2^{m'}<\epsilon]=\lceil 2^{m'} \epsilon\rceil /2^{m'}$ for all $i$ and $m'$. That means that we can use an $\mathbf{AC}^0$ %circuit to compute $W'_1,\cdots,W'_{(nk-k)/(k-1)}$ that are each independently set to $1$ with a probability within $2^{-m'}$ of $\epsilon$ from $W$. %Then, we can plug $X^{(0)}$ and $W'$ into $f$. So, for any constant $c$, there exists an $\mathbf{AC}^0$ circuit that takes $X^{(0)}$ and a list of boolean %variables that are independently set to $1$ with probability $1/2$ as input, and has an output probability distribution that is within $2^{-n^c}$ %of the probability distribution of $X^{(d)}$.
%\end{remark}

Now that we have established that $\mathbf{AC}^0$ circuits are capable of drawing strings from the correct probability distribution, the logical next question is whether or not $\mathbf{NC}^0$ circuits can do the same. As it turns out, they generally cannot. The key issue is that each bit output by an $\mathbf{NC}^0$ circuit is affected by a constant number of input bits. 

\begin{theorem}\label{thm:nc0gen}
Let $f_n:\{0,1\}^{m_n}\to\{0,1\}^{L_d}$ be a series of functions that can be computed by an $\mathbf{NC}^0$ circuit. Also, let $W_1,\cdots,W_{m_n}$ be independently generated random variables and $X'=f_n(W)$. If $0<\theta<1$ then
\[\sum_{x\in\{0,1\}^{L_d}} \min\left(\mathbb{P}[X^{(d)}=x],\mathbb{P}[X'=x]\right)=O\left(e^{-\sqrt{n}}\right)\]
\end{theorem}

We defer the proof of this theorem to Appendix~\ref{app:nc0}. It turns out to be much easier to prove the simpler result that $\mathbf{NC}^0$ fails when it is given uniformly random bits as input, just because some pairs of bits in the output of the broadcast tree model have weak but non-zero correlations.

\bibliographystyle{plain}
\bibliography{all2,my}

\appendix

\section{Lower bounds against $\mathbf{NC}^0$}\label{app:nc0}

If the random bits are each set to $1$ with probability $1/2$, then this means that the probability that any pair of outputs take on any two values must be an integer multiple of $2^{-2c}$ where $c$ is the largest number of input bits affecting a single output bit. However, some of the elements of $X^{(d)}$ have correlations that are less than $2^{-2c}$, so this cannot get the probability distribution right. More formally, we have the following.

\begin{lemma}
Let $f:\{0,1\}^m\to\{0,1\}^n$ be a function that can be computed by an $\mathbf{NC}^0$ circuit,  $W_1,\cdots,W_m$ be independent random variables 
 that are set to $1$ with probability $1/2$ and $0$ with probability $1/2$, and $X'=f(W)$. If $0<\theta<1$ then
 \[
 d_{TV}(X',X^{(d)}) = \Omega(1).
 \]
%\[\sum_{x\in\{0,1\}^n} %\min\left(\mathbb{P}[X^{(d)}=x],\mathbb{P}[X'=x]\right)=1-\Omega(1)\]
\end{lemma}

\begin{proof}
There must exist a constant $c$ such that each output of $f$ is affected by at most $c$ of its inputs. Now, let $d'$ be the smallest positive integer such that $\theta^{2d'}\le 2^{-2c}$. For any $n\ge k^{d'}$, 
\[\mathbb{P}[X^{(d)}_{1^d}=X^{(d)}_{1^{d-d'}k^{d'}}]=1/2+\theta^{2d'}/2\]
However, $\mathbb{P}[X'_{1^d}=X'_{1^{d-d'}k^{d'}}]$ must be an integer multiple of $2^{-2c}$. So, it must be the case that
\[|\mathbb{P}[X^{(d)}_{1^d}=X^{(d)}_{1^{d-d'}k^{d'}}]-\mathbb{P}[X'_{1^d}=X'_{1^{d-d'}k^{d'}}]|\ge \theta^{2d'}/2\]
The desired conclusion follows.
\end{proof}

%\note{Is $1^{d-d'}k^{d'}$ a good way to denote a string of $d-d'$ 1s followed by $d'$ {\em k}s?}

This lemma is somewhat unsatisfying in that it leaves open the possibility that the generation process might be doable in $\mathbf{NC}^0$ if we are given access to independent bits with any desired given biases. We study this case next. 

%The proof above raises the following question. Do generation functions exist if we allow access to independent bits with any desired given biases? This more general case is defined and discussed in *** where we show that there are no generation functions in $\mathbf{NC}^0$ even if we assume that the input $x$ is sampled from a product distribution with any desired biases.

To prove lower bounds against  $\mathbf{NC}^0$ in this more general setup, we will still use the property that each bit output by the $\mathbf{NC}^0$ circuit is affected by a constant number of input bits. Also, each input bit could effect anywhere from $1$ output bit to all of them. That means that if we divide the interval $[1,n]$ into a sufficiently large collection of subintervals, there must be at least one, $[a,b]$, such that less than half of the outputs of the circuit are affected by an input that affects a number of outputs in that range. Then, we can find a set of $\Omega(n/a)$ outputs that only have dependencies as a result of inputs that affect more than $b$ outputs. That allows us to show that for any fixed assignment of values to those inputs the overlap between the probability distributions of $X^{(d)}$ and the output of the circuit is very small. Then, we can add together these overlaps for every assignment of values to those variables and show that it is still small because there are at most $\Omega(n/b)$ inputs that affect that many outputs. Our first step towards proving this will be to show that any $\mathbf{NC}^0$ circuit with a large number of outputs has a large subset of its outputs that are independent conditioned on the values of a relatively small number of inputs. More formally, we have:

\begin{lemma}
Let $f_n:\{0,1\}^{m_n}\to\{0,1\}^n$ be a series of functions that can be computed by an $\mathbf{NC}^0$ circuit, and $c$ be the maximum number of inputs that any output is affected by. Also, let $W_1,\cdots,W_{m_n}$ be independently generated random variables and $X'=f_n(W)$. Next, let $n\ge b_0\ge b_1\ge b_2\ge\cdots\ge b_{2c}\ge 1$. For any given $n$, there exists $0< i\le 2c$, $S\subseteq\{1,\cdots,n\}$ and $T\subseteq\{1,\cdots,m\}$ such that $|S|\ge \frac{n}{2cb_{i}}$, $|T|\le cn/b_{i-1}$, and $\{X'_j: j\in S\}$ are independent conditioned on any fixed value of $\{W_j:j\in T\}$.
\end{lemma}

\begin{proof}
Choose an $n$, refer to $m_n$ as $m$, and for each $j$, let $s_j$ be the number of bits in the output of $f_n$ that are affected by the value of $W_j$. Also, assume without loss of generality that $s_1\ge s_2\ge\cdots\ge s_{m}$. Next, for each $0\le i\le 2c$, let $j_i$ be the smallest positive integer such that $s_{j_i}\le b_i$, or $m+1$ if $s_j> b_i$ for all $j$. Now, observe that
$$\sum_{i=1}^{2c}\sum_{j=j_{i-1}}^{j_i-1} s_j =\sum_{j=j_0}^{j_{2c}-1} s_j \le \sum_{j=1}^m s_j \le cn $$
So, there must exist $i$ such that $\sum_{j=j_{i-1}}^{j_i-1} s_j\le n/2$. That means that there are at least $n/2$ elements of $X$ that are not affected by $W_j$ for any $j_{i-1}\le j<j_i$. For any such element of $X$, there are at most $c(b_i-1)$ other elements of $X$ that are affected by any of the elements of $W_{j_i},W_{j_i+1},\cdots,W_m$ that affect it. So, we can find at least $n/2cb_i$ elements of $X$ such that for all $j_{i-1}\le j\le j_i-1$, $W_j$ does not affect any of them, and for all $j\ge j_i$, at most one of these elements is affected by $W_j$. Also, $j_{i-1}\le cn/b_{i-1}+1$. So, that leaves at most $cn/b_{i-1}$ elements of $W$ that affect more than one of these elements of $X$.
\end{proof}

That establishes that the output of any such $\mathbf{NC}^0$ circuit contains a large number of elements that are independent conditioned on the values of a relatively small number of inputs. Ultimately, we will want to show that the probability distribution of the corresponding elements of $X^{(d)}$ must have negligible overlap with the probability distribution of these outputs. In order to do this, we will need to establish that the probability distribution of any large subset of the elements of $X^{(d)}$ has very low overlap with the probability distribution of any set of independent random variables. The main idea behind that argument is that elements of $X^{(d)}$ corresponding to nearby leaves are correlated. So, any two independent random variables corresponding to nearby leaves must either be excessively biased towards one label or have too low a probability of being equal to each other. As such, we state the following result:

\begin{lemma}
For any fixed values of $0<\theta<1$ and $k>1$, there exist constants $c_1,c_2>0$ such that the following holds. Let $S\subseteq L_d$, and let $X'_i\in\{0,1\}$ be a random variable for each $i\in L_d$ such that $\{X'_i:i\in S\}$ are independent. Then
\[\sum_{x\in\{0,1\}^S} \min\left(\mathbb{P}\left[X^{(d)}_i=x_i\text{ for } i\in S\right],\mathbb{P}\left[X'_i=x_i\text{ for } i\in S\right]\right)\le 2e^{-c_1 |S|^{1+c_2}/n^{c_2}}\]
\end{lemma}

\begin{proof}
First, let $d'=\lfloor \log_k(|S|/6)\rfloor$. Next, let $\delta=\theta^{d-d'}/4$ and $d''=d'-\lceil -\log(4)/\log\theta\rceil$. We will break up our analysis into two cases. 
 
First consider the case where $E[X'_i]\ge 1/2+\delta$ for at least $1/3$ of the $i$ in $S$ or $E[X'_i]\le 1/2-\delta$ for at least $1/3$ of the $i$ in $S$. Assume without loss of generality that $E[X'_i]\ge 1/2+\delta$ for at least $1/3$ of the $i$ in $S$. In this case, let $S'=\{i\in S: E[X'_i]\ge 1/2+\delta\}$. Next, let $S''$ be a maximal subset of $S$ such that $\parent^{(d-d'')}(i)\ne\parent^{(d-d'')}(i')$ for all distinct $i,i'\in S''$. Clearly, there are at most $k^{d-d''}$ elements of $S'$ that have any given ancestor in $L_{d''}$, so $|S''|\ge |S'|/k^{d-d''}\ge |S|/3k^{d-d''}$. Also, $E[X'_i]\ge 1/2+\delta$ for every $i\in S''$. However, for any $x\in\{0,1\}^{L_{d''}}$ and any $i\in S''$, it must be the case that
\begin{align*}
E[X^{(d)}_i|X^{(d'')}=x] &\le 1/2+\theta^{d-d''}/2\\
&\le 1/2+\theta^{d-d'}/8 =1/2+\delta/2
\end{align*}
Also, these elements of $X^{(d)}$ are independent conditioned on any value of $X^{(d'')}$ because $S''$ does not contain the indices of any pair of vertices with a common ancestor closer than $X^{(d'')}$. So, by a Chernoff bound,
\[P\left [\frac{1}{|S''|}\sum_{i\in S''} X_i^{(d)}\ge 1/2+3\delta/4\right]\le e^{-\delta^2|S''|/96}\]
On the flip side,
\[P\left [\frac{1}{|S''|}\sum_{i\in S''} X'_i\le 1/2+3\delta/4\right]\le e^{-\delta^2|S''|/64}\]
So, the overlap between the probability distributions of $X'$ and $X^{(d)}$ is at most $2 e^{-\delta^2|S''|/96}$. Now, observe that
\begin{align*}
\delta^2|S''| &\ge \theta^{2(d-d')}|S|/48k^{d-d''}\\
&\ge \theta^{2(d-d')}|S|/48k^{d-d'}k^{1+\log(4)/\log\theta}\\
&\ge \theta^2[|S|/6n]^{1-2\log_k\theta}|S|/48k^{2+\log(4)/\log\theta}\\
&=\frac{\theta^2}{48\cdot 6^{1-2\log_k\theta}k^{2+\log(4)/\log\theta}}\cdot\frac{|S|^{2-2\log_k\theta}}{n^{1-2\log_k\theta}}
\end{align*}
So, the overlap between the probability distributions of $X'$ and $X^{(d)}$ is at most
\[2e^{-\frac{\theta^2}{4608\cdot 6^{1-2\log_k\theta}k^{2+\log(4)/\log\theta}}\cdot\frac{|S|^{2-2\log_k\theta}}{n^{1-2\log_k\theta}}}\]

Now we consider the remaining case when $1/2-\delta\le E[X'_i]\le 1/2+\delta$ for at least $1/3$ of the $i$ in $S$. In this case, let $S'=\{i\in S:1/2-\delta\le E[X'_i]\le 1/2+\delta\}$. We know that $|S'|\ge |S|/3\ge 2k^{d'}$. So, there must be at least $(|S'|-k^{d'})/k^{d-d'}\ge |S|/6k^{d-d'}$ values of $j\in L_{d'}$ such that more than one of the elements of $L_{d-d'}(j)$ are in $S'$. Now, pick $i,i'\in L_{d-d'}(j)\cap S'$ for each such $j$, and let $S''$ be the set of all such pairs $(i,i')$. For any such $i,i'$, we know that $X'_i$ is independent of $X'_{i'}$, so $\mathbb{P}[X'_i=X'_{i'}]\le 1/2+2\delta^2$. Also, $\mathbb{P}[X^{(d)}_i=X^{(d)}_{i'}]\ge 1/2+\theta^{2d-2d'}/2=1/2+8\delta^2$, and this probability is independent of the labels of any leaves not descended from $\parent^{(d-d')}(i)$. So, by a Chernoff bound, 
\[\mathbb{P}[|\{(i,i')\in S'': X^{(d)}_i=X^{(d)}_{i'}\}|/|S''|\le 1/2+5\delta^2]\le e^{-9\delta^4|S''|/4}\]
and
\[\mathbb{P}[|\{(i,i')\in S'': X'_i=X'_{i'}\}|/|S''|\ge 1/2+5\delta^2]\le e^{-9\delta^4|S''|/6}\]
So, the overlap between the probability distributions of $X'$ and $X^{(d)}$ is at most $2e^{-9\delta^4|S''|/6}$. Now, observe that
\begin{align*}
9\delta^4|S''|/6 &\ge \delta^4 |S|/4k^{d-d'}\\
&= \theta^{4(d-d')}|S|/1024k^{d-d'}\\
&\ge \frac{\theta^4}{1024k}\cdot |S|(|S|/6n)^{1-4\log_k\theta}
\end{align*}
Thus, the overlap between the probability distributions is at most $2e^{-\frac{\theta^4}{1024k}\cdot |S|(|S|/6n)^{1-4\log_k\theta}}$. 

\vspace{5mm}
So, the desired conclusion holds with 
\[c_1=\min\left(\frac{\theta^2}{4608\cdot 6^{1-2\log_k\theta}k^{2+\log(4)/\log\theta}},\frac{\theta^4}{1024k\cdot 6^{1-4\log_k\theta}}\right)\]
 and $c_2=1-4\log_k\theta$.
\end{proof}

 So, at this point we have established that any $\mathbf{NC}^0$ circuit with independent random inputs must have a large set of outputs that are independent conditioned on any assignment of values to a relatively small set of inputs. Also, we know that the overlap between the probability distribution of the outputs conditioned on an assignment of value to these inputs and the probability distribution of $X^{(d)}$ must be small. Now, we just need to add up the overlaps for every possible assignment of values to these inputs in order to bound the overall overlap between the probability distribution of $X^{(d)}$ and the probability distribution of the circuit's output. 
 
 We are now ready to prove Theorem~\ref{thm:nc0gen}:
 
 \begin{proof}
First, let $c$ be the maximum number of inputs that any output of $f$ is ever affected by. Next, for each integer $0\le i\le 2c$, let $b_i=e^{\ln(n)^{(2c-i)/2c}}$. There must exist $0<i\le 2c$, $S\subseteq\L_d$ and $T\subseteq\{1,\cdots,m_n\}$ such that $|S|\ge\frac{n}{2cb_i}$, $|T|\le cn/b_{i-1}$, and $\{X'_j: j\in S\}$ are independent conditioned on any fixed value of $\{W_j: j\in T\}$. Now, choose $c_1$ and $c_2$ satisfying the conditions of the previous lemma. Then, for every $w\in\{0,1\}^{|T|}$, let $E_w$ be the event that the elements of $W$ with indices in $T$ take on the values given by $w$. Observe that
\begin{align*}
&\sum_{x\in\{0,1\}^{L_d}} \min\left(\mathbb{P}[X^{(d)}=x],\mathbb{P}[X'=x]\right)\\
&\le \sum_{x\in\{0,1\}^{L_d}} \sum_{w\in\{0,1\}^{|T|}} \min\left(\mathbb{P}[X^{(d)}=x],\mathbb{P}[X'=x,E_w]\right)\\
&\le \sum_{w\in\{0,1\}^{|T|}} \sum_{x\in\{0,1\}^{L_d}} \min\left(\mathbb{P}[X^{(d)}=x],\mathbb{P}[X'=x|E_w]\right)\\
&\le \sum_{w\in\{0,1\}^{|T|}} 2e^{-c_1|S|^{1+c_2}/n^{c_2}}\\
&= 2^{|T|+1}e^{-c_1|S|^{1+c_2}/n^{c_2}}\\
&\le 2^{cn/b_{i-1}+1}e^{-c_1n/(2cb_i)^{c_2}}\\
&=2 e^{\ln(2)cn/b_{i-1}-c_1n/(2cb_i)^{c_2}}\\
\end{align*}

Also, $b_i^{c_2}=o(b_{i-1})$. So, there exists $n_0$ such that for all $n\ge n_0$ and all integers $0<i\le 2c$, we have that $\ln(2)cn/b_{i-1}\le c_1n/(2cb_i)^{c_2}/2$.That means that for all $n\ge n_0$,
\begin{align*}
&\sum_{x\in\{0,1\}^{L_d}} \min\left(\mathbb{P}[X^{(d)}=x],\mathbb{P}[X'=x]\right)\\
&\le 2 e^{-c_1n/(2cb_i)^{c_2}/2}\\
&\le 2 e^{-c_1n/(2cb_1)^{c_2}/2}\\
\end{align*}

$\ln(b_1)=\ln(n)^{1-1/2c}=o(\ln(n))$, so 
\[\sum_{x\in\{0,1\}^{L_d}} \min\left(\mathbb{P}[X^{(d)}=x],\mathbb{P}[X'=x]\right)=O\left(e^{-\sqrt{n}}\right)\]
as desired.
\end{proof}

\section{Random restrictions in the broadcast tree model}\label{app:ac0}

 Here we prove Theorem~\ref{thm:ac0main}. The usual approach to proving lower bounds against $\mathbf{AC}^0$ is through {\em random restrictions} where for every input $x_i$ we leave it unset with probability $p$ and otherwise we set it to zero with probability $\frac{1-p}{2}$ and set it to one with the remaining probability $\frac{1-p}{2}$. The main insight is that if the parameters are chosen appropriately, with high probability the $\mathbf{AC}^0$ circuit becomes much simpler (while the parity function remains a parity on fewer inputs). The key to our lower bound is an alternative but equivalent way to generate samples from the broadcast tree model. We will need the following definitions:

\begin{definition}
Let $d'>0$. Let $\phi_{d'}: \{0,1\}^{L_{d'-1}} \times \{0,1,*\}^{L_{d'}} \rightarrow \{0,1\}^{L_{d'}}$ be the function such that for all $x\in\{0,1\}^{L_{d'-1}}, r \in \{0,1,*\}^{L_{d'}}$ and $v \in L_{d'}$, we have that
\[(\phi_{d'}(x,r))_v=
\begin{cases}
r_v  &\text{ for } r_v\in\{0,1\}\\
x_{\parent(v)} &\text{ for } r_i= *, 
\end{cases}
\]
\end{definition}

For the tree broadcast process the natural distribution for $r$ is given by independent copies of the following distribution: 

\begin{definition}
For any $0\le \theta<1$, let $R_\theta$ be the probability distribution over $\{0,1,*\}$ such that a variable drawn from $R_\theta$ will be $0$ with probability $(1-\theta)/2$, $1$ with probability $(1-\theta)/2$ , and '*' with probability $\theta$.
\end{definition}

\begin{definition}
Let $d'>0$. Let $\Phi_{d'}: \{0,1\}^{L_{d'-1}}  \rightarrow \{0,1\}^{L_{d'}}$ be the random function such that for all $x\in\{0,1\}^{L_{d'-1}}$ and $v \in L_{d'}$, we let $\Phi_{d'}(x) = \phi_{d'}(x,r)$, 
where $r$ is drawn from $R_{\theta}^{L_{d'}}$
\end{definition}

One can easily check that for all $k$, $\theta$, and $d'$ and all $x\in \{0,1\}^{L_{d'-1}}$ and $x'\in\{0,1\}^{L_{d'}}$, if we have $r\sim R_\theta^{L_{d'}}$ then
\[
\mathbb{P}[\Phi(x)=x'] = \mathbb{P}[\phi(x,r)=x'] = \mathbb{P}[X^{(d')}=x'|X^{(d'-1)}=x].
\]

 For any $x\in\{0,1\}$, the probability distribution of $\Phi_{d}\circ\Phi_{d-1}\cdots\circ\Phi_{1}(x)$ is identical to the probability distribution of $X^{(d)}$ given that $X^{(0)}=x$. So as in classical applications of the switching lemma, we want to show that for any $f\in \mathbf{AC}^0$, $f\circ \Phi_{d}\circ\Phi_{{d-1}}\cdots\circ\Phi_{1}$ is a constant function with high probability. The first step towards doing that will be to prove that applying a logarithmic number of these projections to an $\mathbf{AC}^0$ circuit is enough to reduce the fan-in of all gates in its bottom layer to a constant with high probability. For that, we will need the following.

%{\em EM: Notation is a bit inconsistent. I like $\Phi_d$ better than $\Phi_{r_d}$ below}
\begin{lemma}\label{restrictLem}
Let $m$, $d'$, $h$, and $c$ be positive integers. Also, let $f:\{0,1\}^{L_{d'}}\rightarrow \{0,1\}$ be a function such that there are only $m$ inputs that ever affect its value. 
Next, let $f'=f\circ \Phi_{d'}\circ \Phi_{d'-1}\circ\cdots\circ \Phi_{d'-h+1}$
%Next, let $f'=f\circ \Phi_{r_{d'}}\circ \Phi_{r_{d'-1}}\circ\cdots\circ \Phi_{r_{d'-h+1}}$. 
With probability at least $1-(m\theta^{h})^c$, there are fewer than $c$ inputs that affect the value of $f'$.
\end{lemma}

\begin{proof}
Each time we compose the function with $\Phi_{i}$, each of its inputs is independently set to a constant with probability $1-\theta$, and then some of the inputs might be set to the same variable. If $f'$ depends on $c$ or more inputs, then there must be a set of $c$ of the inputs of $f$ that affect it such that none of these inputs get set to a constant or set to the same variable by any of the projections. There are at most $m^c$ sets of $c$ inputs of $f$ that affect its value, and for any such set, the probability that none of them get set to a constant or merged is at most $\theta^{ch}$. The desired conclusion follows.
\end{proof}

\begin{corollary}
Let $b$ and $h'$ be positive constants, $d'>h'\ln(n)$ be a positive integer, and $f:\{0,1\}^{L_{d'}}\rightarrow \{0,1\}$ be a function that takes an AND or OR of some subset of its inputs and their negations. %Also, let $f'=f\circ \Phi_{r_{d'}}\circ \Phi_{r_{d'-1}}\circ\cdots\circ \Phi_{r_{d'-\lceil h'\ln(n)\rceil+1}}$.
Also, let $f'=f\circ \Phi_{d'}\circ \Phi_{d'-1}\circ\cdots\circ \Phi_{d'-\lceil h'\ln(n)\rceil+1}$.
With probability at least $1-O(n^{-b})$, $f'$ is an AND or OR of $-b/h'\ln\theta+1$ or fewer inputs and negations of inputs.
\end{corollary}

\begin{proof}
First of all, observe that if $f$ takes an AND/OR of more than $2b\ln(n)$ variables, one of the projections will set one of its inputs to the value that reduces the function to a constant with probability $1-o(n^{-b})$. Otherwise, the desired conclusion follows by the lemma and the fact that a projection of an AND or OR must still be an AND or OR.
\end{proof}

Now that we know that applying a logarithmic number of projections to the circuit will reduce the fan-in of all gates in the bottom layer to a constant with high probability, our next step is to prove that one more projection is enough to reduce all gates on the second layer to decision trees of logarithmic depth. In order to do that, we will need to prove that the projection of one of these gates can be represented by a decision tree of height $O(\log(n))$ with probability $1-n^{\Omega(1)}$. In other words, we need:

\begin{lemma}\label{decisionTree}
Let $w$ be a positive constant. There exists a constant $h>0$ such that if $p>0$ is a function of $n$ and $f$ is a $w$-DNF on $\{0,1\}^{L_{d'}}$ then $f\circ\Phi_{d'}$ can be represented as a decision tree of height at most $h\ln(p)$ with probability $1-O(1/p)$.
\end{lemma}

\begin{proof}
We proceed by induction on $w$. If $w=0$, then every $w$-DNF is a constant function, and is thus expressible as a decision tree of height $0$. Now, assume this result holds for $w-1$. If $f$ is a $w$-DNF with more than $(\frac{1-\theta}{2})^{-w}\ln(p)$ clauses that do not share any variables, then with probability $1-O(1/p)$, the projection sets at least one of these clauses to $1$, with the result that $f$ becomes a constant function. Otherwise, there exists a set of at most $$w\cdot \left(\frac{1-\theta}{2}\right)^{-w}\ln(p)$$ variables such that at least one of these variables appears in every clause. As such, any assignment of values to these variables would reduce $f$ to a $(w-1)$-DNF. By the induction hypothesis, there exists a constant $h'$ such that composing any resulting $(w-1)$-DNF with $\Phi_{d'}$ yields a decision tree of height at most $h'\ln(p)$ with probability at least $$1-O(p^{-1-\ln(2)w\cdot(\frac{1-\theta}{2})^{-w}})$$ That means that all assignments of values to these variables reduce $f\circ\Phi_{d'}$ to a decision tree of depth $h'\ln(p)$ with probability at least $1-O(1/p)$. Therefore, $f\circ\Phi_{d'}$ can be represented as a decision tree of depth 
$$\left[w\cdot \left(\frac{1-\theta}{2}\right)^{-w}+h'\right]\ln(p)$$ with probability $1-O(1/p)$. This completes the proof.
\end{proof}

%\fbox{Note: In this lemma $h\approx(\frac{1-\theta}{2})^{-w^2/2}$.}

At this point, we know that applying a logarithmic number of projections is enough to reduce every gate on the second level of an $\mathbf{AC}^0$ circuit to a decision tree of logarithmic depth with high probability. Any such decision tree can be computed by a polynomial size AND of ORs and by a polynomial size OR of ANDs. So, we can replace it by whichever allows us to reduce the circuit depth by 1. We are applying $\Omega(\log(n))$ projections in total, so if the circuit has depth $b$ we can divide the projections into $b$ serieses of $\Omega(\log(n))$ projections each. That is enough to reduce the entire circuit to a decision tree of logarithmic depth with a logarithmic number of projections left over. In fact, we can prove that with high probability, the depth of the decision tree is low enough that it must be unaffected by the values of most of the variables. Then, we can show that the remaining projections set all of the variables that the output does depend on to constants with high probability. As such, we can prove:

\begin{lemma}
Let $f: \{0,1\}^{L_d}\rightarrow \{0,1\}$ be in $\mathbf{AC}^0$. Then there exists $\delta>0$ such that $f\circ \Phi_{d}\cdots\circ \Phi_{1}$ is a constant function with probability $1-O(n^{-\delta})$.
\end{lemma}

\begin{proof}
First, let $f^{(0)}=f$ and $f^{(i+1)}=f^{(i)}\circ\Phi_{d-i}$ for each $i$. Also, let $b$ be the depth of $f$, and $\delta_1>0$ be a constant. We claim that $f^{(\lfloor i d/b\rfloor)}$ can be expressed as a polynomial-sized circuit of depth $b-i$ with probability $1-O(n^{-\delta_1})$ for each $0\le i<b-1$, and prove this by induction on $i$. This is clearly true for $i=0$. For $i>0$, if $f^{(\lfloor (i-1) d/b\rfloor)}$ can be expressed as a polynomial-sized circuit of depth $b-i+1$, then by corollary $1$ there exists a constant $c_i$ such that $f^{(\lfloor i d/b\rfloor-1)}$ can be expressed as a polynomial-sized circuit of depth $b-i+1$ in which every gate at the bottom level has fanin at most $c_i$ with probability $1-O(n^{-\delta_1})$. Then by lemma \ref{decisionTree}, composing this with $\Phi_{d-\lfloor i d/b\rfloor+1}$ allows us to replace all gates two levels from the bottom with decision trees of depth $O(\ln(n))$ with probability $1-O(n^{-\delta_1})$. Every such decision tree can be converted to a DNF or CNF of size polynomial in $n$, so we can apply this transformation to all such gates in order to switch the order of the ORs and ANDs, thus allowing us to reduce the depth of the circuit by $1$. Thus, $f^{(\lfloor i d/b\rfloor)}$ can be expressed as a polynomial-sized circuit of depth $b-i$ with probability $1-O(n^{-\delta_1})$, as desired.

That leaves us with the conclusion that $f^{(\lfloor(b-2)i d/b\rfloor)}$ can be expressed as a polynomial-sized circuit of depth $2$ with probability $1-O(n^{-\delta_1})$. Then, by another application of corollary $1$, we have that $f^{(\lfloor(b-1)i d/b\rfloor-1)}$ can be expressed as a DNF or CNF of constant fanin with probability $1-O(n^{-\delta_1})$. Then, by lemma \ref{decisionTree}, we have that there exists a constant $h$ such that $f^{(\lfloor(b-1)i d/b\rfloor)}$ can be expressed as a decision tree of depth $h\delta_2 \log_2(n)$ with probability $1-O(n^{-\delta_2})$ for any $\delta_2>0$. Such a decision tree can only be affected by $n^{h\delta_2}$ variables, so by lemma \ref{restrictLem}, $f^{(d)}$ is a constant function with probability $1-O(n^{-\delta_1}+n^{-\delta_2}+\theta^{d/b}n^{h\delta_2})$. For $\delta_2<-\ln\theta/[b(h+1)\ln(k)]$ and $\delta\le\min(\delta_1,\delta_2)$, that means that $f^{(d)}$ is a constant with probability $1-O(n^{-\delta})$.
\end{proof}

Recall that for any fixed value of $X^{(0)}$, the probability distribution of $\Phi_{d}\cdots\circ \Phi_{1}(X^{(0)})$ is identical to the probability distribution of $X^{(d)}$. So, the probability that $f(X^{(d)})=X^{(0)}$ is the same as the probability that $f\circ \Phi_{d}\cdots\circ \Phi_{1}(X^{(0)})=X^{(0)}$. That means that the fact that $f\circ \Phi_{d}\cdots\circ \Phi_{1}$ is probably a constant implies that $f$ is only accurate about half of the time. 

Now we are ready to prove Theorem~\ref{thm:ac0main}:

\begin{proof}
Let $(r_1,\cdots,r_d)$ be bad if $f\circ \Phi_{d}\cdots\circ \Phi_{1}$ is a constant function and good otherwise. Then we have that
\begin{align*}
\mathbb{P}[f(X^{(d)})=X^{(0)}] &=\mathbb{P}[f\circ \Phi_{d}\cdots\circ \Phi_{1}(X^{(0)})=X^{(0)}]\\
&\le 1/2+\mathbb{P}[(r_1,\cdots,r_d) \text{ is good}]/2\\
&=1/2+O(n^{-\delta})
\end{align*}
which completes the proof.
\end{proof}

\begin{corollary}
For every $c>0$, there is no function in $\mathbf{AC}^0$ that computes whether more than half of its inputs are $1$ whenever at least $n/2+n^{1-c}/2$ of its inputs are the same.
\end{corollary}

\begin{proof}
For any such $c$, there is a choice of $0\le\theta<1$ and $k>0$ such that more than $n/2+n^{1-c}/2$ of the entries in $X^{(d)}$ will equal $X^{(0)}$ with a probability of at least $2/3$. So, any such function would be capable of computing $X^{(0)}$ from $X^{(d)}$ with nontrivial accuracy.
\end{proof}

\section{Computing the posterior in $\mathbf{NC}^1$}\label{app:nc1part1}
 
 Here we prove that there is an $\mathbf{NC}^1$ circuit for computing the posterior. This is the first part of Theorem~\ref{thm:nc1main}.
 
 More formally, we claim the following.

\begin{lemma}
For every $-1<\theta<1$ and positive integer $k$, there exists $h>0$ such that for every $d'\ge 0$, there exists a probability distribution $P_{F}^{d'}$ over functions from $\{0,1\}^{L_{d'}}\rightarrow \{0,1,?\}$ with the following properties: 
\begin{itemize}
    \item Every function drawn from $P_{F}^{d'}$ can be computed by an NC circuit of depth at most $hd'$.
    \item For every $x\in \{0,1\}^{L_{d'}}$, if $F\sim P_{F}^{d'}$ then $\mathbb{P}[F(x)\in \{0,1\}]\ge 1-1/2k$.
    \item $\mathbb{P}[F(x)=1|F(x)\in\{0,1\}]=\mathbb{P}[X^{(0)}=1|X^{(d')}=x]$
\end{itemize}

\end{lemma}

\begin{proof}
We proceed by induction on $d'$. For $d'=0$, we can always return the function $f$ such that $f(0)=0$ and $f(1)=1$. Now, assume that this holds for $d'-1$. Prior to defining $P_{F}^{d'}$, we will define a preliminary probability distribution $P_{F}^{d'\star}$, such that in order to draw a function $F^\star$ from $P_{F}^{d'\star}$, we do the following. First, draw $F_1,\cdots,F_k$ independently from $P_{F}^{d'-1}$. Also, independently choose $\delta_1,\cdots,\delta_k$ such that for each $i$, $\delta_i$ is $1$ with probability $(1-\theta)/2$ and $0$ otherwise. 

If there exists $i$ such that $F_i(x(L_{d'-1}(i)))
\not\in\{0,1\}$, then $F^\star(x)='?'$. Otherwise, let $x^\star_i=F_i(x(L_{d'-1}(i)))$ for each $i$. Then, set $F^\star(x)$ equal to $0$ if $x^\star_i\text{  xor }\delta_i=0$ for all $i$, set it to $1$ if  $x^\star_i\text{  xor } \delta_i=1$ for all $i$, and set it to $'?'$ otherwise.
%If there exists $i$ such that $F_i(x_{(i-1) k^{d'-1}+1},\cdots,x_{i k^{d'}})\not\in\{0,1\}$, then $F^\star(x)='?'$. Otherwise, let $x^\star_i=F_i(x_{(i-1) k^{d'-1}+1},\cdots,x_{i k^{d'}})$ for each $i$. Then, set $F^\star(x)$ equal to $0$ if $x^\star_i\text{  xor }\delta_i=0$ for all $i$, set it to $1$ if  $x^\star_i\text{  xor } \delta_i=1$ for all $i$, and set it to $'?'$ otherwise

For any fixed value of $x$, when $F^\star\sim P_{F}^{d'\star}$, the values of the $x_i^\star$ are independent. As such, 
\[\mathbb{P}[F^\star(x)=0]=\prod_{i=1}^k \left[\left(\frac{1+\theta}{2}\right)\mathbb{P}[x^\star_i=0]+\left(\frac{1-\theta}{2}\right) \mathbb{P}[x^\star_i=1]\right]\]
and
\[\mathbb{P}[F^\star(x)=1]=\prod_{i=1}^k \left[\left(\frac{1+\theta}{2}\right)\mathbb{P}[x^\star_i=1]+\left(\frac{1-\theta}{2}\right) \mathbb{P}[x^\star_i=0]\right]\]
By the requirement that $\mathbb{P}[F_i(x)\in \{0,1\}]\ge 1-1/2k$, the $x^\star$ are all assigned values in $\{0,1\}$ with probability at least $1/2$, which also implies that 
\[\mathbb{P}[F^\star(x)\in\{0,1\}]\ge [(1-\theta^2)/4]^{k/2}\]
%Now, if we compute $F^\star(X^{(d')})$, then for every $x\in\{0,1\}^{k^{d'}}$ and each $i$, it will be the case that 
By the induction hypothesis, 
\[\mathbb{P}[x^\star_i=1|x^\star_i\ne '?' ,X^{(d')}=x]=\mathbb{P}[X^{(1)}_i=1|(X_{(i-1) k^{d'-1}+1},\cdots,X_{i k^{d'}})=(x_{(i-1) k^{d'-1}+1},\cdots,x_{i k^{d'}})]\] 
This implies that
\[\mathbb{P}[F^\star(x)=1|F^\star(X^{(d')})\in\{0,1\}]=\mathbb{P}[X^{(0)}=1|X^{(d')}=x]\]
Furthermore, $F^\star(x)$ can be computed from the values output by the $F_i$ by an NC circuit of some constant depth. 

So, $P_{F}^{d'\star}$ has all of the properties that we want, except that its functions return $'?'$ with excessively high probability. So, in order to draw a function from $P_{F}^{d'}$, we simply draw $\lceil[(1-\theta^2)/4]^{-k/2}\ln(2k)\rceil$ independent functions from $P_{F}^{d'\star}$. Then, we compute them all on the input we are given, and return the first output in $\{0,1\}$, if any. This leaves the relative probability of returning $0$ and $1$ unchanged, reduces the probability of returning $'?'$ to $1/2k$ or less, and only increases the circuit depth by a constant. So, $P_{F}^{d'}$ has all of the desired properties.
\end{proof}

That means that we have randomized $\mathbf{NC}^1$ circuits that can essentially draw a sample from the probability distribution of $X^{(0)}$ given that $X^{(d)}=x$, with the complication that they occasionally fail to return a value in $\{0,1\}$. So, if we use a large number of them, we can count up how many of them return $1$ and how many return $0$ in order to estimate $\mathbb{P}[X^{(0)}=1|X^{(d)}=x]$. With enough of these circuits, this estimate will be within $\delta/2$ of the true probability at least $1-o(2^{-n})$ of the times, which means that there must be some choice of the randomness for which it is always right. That allows us to prove:

\begin{proposition}
For every $k$ and $\theta$ and in the Ising tree model, there is a posterior function that can be computed by an $\mathbf{NC}^1$ circuit. 
\end{proposition}

\begin{proof}
Consider independently drawing $F_1,\cdots,F_{n^4}$ from $P_{F}^{d}$. Also, consider any $x\in\{0,1\}^n$ such that $\mathbb{P}[X^{(0)}=1|X^{(d)}=x]>1/2+1/n$. For each $i$, it is the case that $\mathbb{P}[F_i(x)\in\{0,1\}]\ge 1/2$ and $\mathbb{P}[F_i(x)=1|F_i(x)\in\{0,1\}]\ge 1/2+1/n$. So, there will be more $i$ for which $F_i(x)=1$ than $i$ for which $F_i(x)=0$ with probability $1-o(2^{-n})$. That in turn implies that this holds for every such $x$ with probability $1-o(1)$. By the same logic, there will be more $i$ for which $F_i(x)=0$ than $i$ for which $F_i(x)=1$ for every $x$ such that $\mathbb{P}[X^{(0)}=1|X^{(d)}=x]<1/2-1/n$ with probability $1-o(1)$. That means that there must exist a specific choice of $F_1,\cdots,F_{n^4}$ for which both of these properties hold. These functions can each be computed by an $\mathbf{NC}^1$ circuit, and a logarithmic additional depth is sufficient to determine whether more of them output $1$ or $0$. Thus, the function that returns $1$ if more of them output $1$ than $0$ and $0$ otherwise is an $\mathbf{NC}^1$ posterior function. 

%To prove the theorem we need to consider $x$'s 
%such that $\mathbb{P}[X^{(0)}=1|X^{(d)}=x] > 1/2 + \delta_d$ for $\delta_d \to 0$. For this, we  compare how many of the $F_i$ returned $1$ and $0$ in order to estimate the probability that the root has a label of $1$. For any $c>0$, it would be possible to ensure that our estimate was within $O(n^{-c})$ of the true probability for all possible inputs, provided we used $n^{2c+2}$ functions $F_i$ instead of $n^2$.
\end{proof}

\section{Gadgets in the broadcast tree model}\label{app:nc1part2}

Here we prove that being able to compute the posterior allows us to implement any $\mathbf{NC}^1$ circuit. This is the second part of Theorem~\ref{thm:nc1main}.

\begin{lemma}
Let $\theta=\frac{9}{10}$ and $k=6$. Next, choose $x\in\{0,1\}^{L_d}$ such that there are at least $4$ choices of $i$ such that $\mathbb{P}[X_i^{(1)}=1|X^{(d)}(L_{d-1}(i))=x(L_{d-1`}(i))]\ge.95$. Then $\mathbb{P}[X^{(0)}=1|X^{(d)}=x]\ge \frac{19}{20} $.
\end{lemma}

\begin{proof}
First of all, for each $i$, let $p_i=\mathbb{P}[X_i^{(1)}=1|X^{(d)}(L_{d-1}(i))=x(L_{d-1}(i))]$. Given these values of $\theta$ and $k$, it must be the case that
\begin{align*}
\mathbb{P}[X^{(0)}=1|X^{(d)}=x]&=\frac{\prod_{i=1}^6 \Big (\frac{19}{20} p_i+\frac{1}{20}(1-p_i)\Big )}{\prod_{i=1}^6 \Big(\frac{19}{20} p_i+\frac{1}{20}(1-p_i)\Big )+\prod_{i=1}^6 \Big (\frac{1}{20} p_i+\frac{19}{20}(1-p_i)\Big )}\\
&\ge\frac{\Big ( (\frac{19}{20})^2+(\frac{1}{20})^2\Big )^4\cdot(\frac{1}{20})^2}{\Big ( (\frac{19}{20})^2+(\frac{1}{20})^2\Big )^4\cdot(\frac{1}{20})^2+\Big (\frac{19}{20}\cdot\frac{1}{20}+\frac{1}{20}\cdot\frac{19}{20}\Big )^4\cdot(\frac{19}{20})^2}\\
&>\frac{19}{20}
\end{align*}
\end{proof}

That allows us to prove the following.

\begin{proposition}
Let $\theta=\frac{9}{10}$ and $k=6$ and consider the Ising tree model. For every $\mathbf{NC}$ circuit of depth $d$, there exists a way to define $x\in\{0,1\}^{L_d}$ so that $x_i$ is set to $0$, $1$, an input to the circuit, or the negation of the input to the circuit, such that for every choice of inputs to the circuit, $\mathbb{P}[X^{(0)}=1|X^{(d)}=x]$ is at least $\frac{19}{20}$ if the circuit outputs $1$ on this input, and at most $\frac{1}{20}$ if the circuit outputs $0$ on this input.
\end{proposition}

\begin{proof}
We proceed by induction on $d$. This is clearly true for $d=0$. Now, assume that it holds for $d-1$, and consider a function $f$ that is computable by an $\mathbf{NC}$ circuit of depth $d$. There must exist functions $f_1$ and $f_2$ that are computable by $\mathbf{NC}$ circuits of depth $d-1$ such that either $f=NOT(f_1)$, $f=f_1 \text{ AND } f_2$, or $f=f_1 \text{ OR } f_2$. By the induction hypothesis, for each $i,j$ there is a way to set all of the entries in $\{x_{i'}:i'\in L_{d-1}(i)\}$ equal to constants, inputs to the circuit, or negations of inputs in such a way that $\mathbb{P}[X^{(1)}_i=1|X^{(d)}(L_{d-1}(i))=x(L_{d-1}(i))]$ is always at least $\frac{19}{20}$ if $f_j$ outputs $1$ and at most $\frac{1}{20}$ if it outputs $0$. Also, $\mathbb{P}[X^{(1)}_i=1|X^{(d)}(L_{d-1}(i))=(1,\ldots,1)]>\frac{19}{20}$ and $\mathbb{P}[X^{(1)}_i=1|X^{(d)}(L_{d-1}(i))=(0,\ldots,0)]<\frac{1}{20}$ by repeated application of the previous lemma. 

In particular, if we set $x$ so that $$\mathbb{P}[X^{(1)}_i=1|X^{(d)}(L_{d-1}(i))=x(L_{d-1}(i))]$$ tracks $f_1$ for $i=1,2$, $f_2$ for $i=3,4$, and $1$ for $i=5,6$ then $\mathbb{P}[X^{(0)}=1|X^{(d)}=x]$ will track $f=f_1 \text{ OR } f_2$ by the lemma. If we have it track $0$ for $i=5,6$ instead, then it will track $f=f_1 \text{ AND } f_2$ instead. That leaves the case where $f=NOT(f_1)$. In that case, we can simply start with the assignment of value to $x$ that we would use if $f=f_1=f_1 \text{ OR } f_1$, and then invert every entry in $x$ in order to switch the probability that $X^{(1)}_i=1$ with the probability that it is $0$. So, the desired conclusion holds for $d$.
\end{proof}

\begin{remark}
More generally, given any $\theta, k,\delta>0$ such that $\lim_{d\rightarrow\infty} \mathbb{P}[X^{(0)}=1|X^{(d)}=1,1,\cdots,1]>1/2+\delta$, determining whether $\mathbb{P}[X^{(0)}=1|X^{(d)}=x]>1/2$ whenever this probability is greater than $1/2+\delta$ or less than $1/2-\delta$ is $\mathbf{NC}^1$-hard. This can be proven by a variant of the above argument. In it we would argue that there exists $\delta'$ such that if $f_1$ and $f_2$ are functions that can be tracked by trees of depth $d'$ with an accuracy of $1/2+\delta$, we can construct a tree of depth $d'+2$ that tracks $f_1\text{ AND }f_2$ with accuracy $1/2+\delta'$. Then we would  argue that we can amplify the accuracy back up to $1/2+\delta$ by constructing a tree such that all of its subtrees at some suitable depth are copies of that tree. That would allow us to prove the desired result by induction on circuit depth the same way we do in the theorem above.
\end{remark}

Combining this with the previous theorem shows that posterior computation is $\mathbf{NC}^1$-complete, as desired.

\section{Deviation bounds for the broadcast tree model}\label{app:deviation}

Here we prove Lemma~\ref{lem:deviation}:

 \begin{proof}
First, observe that 
\[E\left[ \sum_{i\in L_d} X^{(d)}_i\middle|X^{(0)}=1\right]=k^d/2+k^d \theta^d/2\]
Now, for each $0\le d'\le d$, let $v_{d'}=Var\left[\sum_{i\in L_{d-d'}(1^{d'})} X^{(d)}_i\middle|X^{(d')}_{1^{d'}}=1\right]$. Clearly, $v_{d}=0$, and for each $d'<d$, it must be the case that
\[v_{d'}=k Var\left[X^{(d'+1)}_{1^{d'+1}}\middle|X^{(d')}_{1^{d'}}=1\right]\cdot k^{2d-2d'-2}\theta^{2d-2d'-2}+k \cdot v_{d'+1}\]
And hence we have
\begin{align*}
v_0&=\sum_{d'=0}^{d-1} \frac{1-\theta^2}{4} k^{d'+1}\cdot  k^{2d-2d'-2}\theta^{2d-2d'-2}\\
&\le \frac{1-\theta^2}{4}\sum_{d'=0}^\infty  k^{2d-d'-1}\theta^{2d-2d'-2}\\
&= \frac{1-\theta^2}{4}k^{2d}\theta^{2d}/[\theta^2k-1]
\end{align*}
In particular, this implies that
\begin{align*}
\mathbb{P}\left[\sum_{i\in L_d} X^{(d)}_i\le k^d/2\middle | X^{(0)}=1\right]&\le \mbox{Var}\left[ \sum_{i\in L_d} X^{(d)}_i\middle|X^{(0)}=1\right]/\left(\mathbb{E}\left[ \sum_{i\in L_d} X^{(d)}_i\middle|X^{(0)}=1\right]-k^d/2\right)^2\\
&=\frac{1-\theta^2}{\theta^2k-1} \le \frac{1}{\theta^2k-1}
\end{align*}
\end{proof}

\section{Reducing the number of labels}\label{app:labelred}

It turns out that we will be able to exploit the symmetries in our generalized broadcast tree model in the previous section to be able to drastically reduce the number of labels from $|A_5|^2 = 3600$ corresponding to all pairs of even permutations to $16$ corresponding to pairs of {\em conjugacy classes} of even permutations in $S_5$ \---- namely two even permutations $\sigma$ and $\tau$ are in the same conjugacy class if there is a permutation $c$ (not necessarily even) for which $\tau = c^{-1} \sigma c$. Intuitively, knowing the conjugacy class fixes the cycle structure of a permutation. 

%While knowing that there is a set of parameters for which determining $X^{(0)}$ with nontrivial accuracy is $\mathbf{NC}^1$-hard is interesting, having $3600$ labels corresponding to ordered pairs of elements of $A_5$ is substantially more complicated than we would really like. As such, our next order of business is to find simpler choices of parameters for which average-case reconstruction is still $\mathbf{NC}^1$-hard. In order to do that, we will prove that we can define the labels to represent pairs of conjugacy classes of even permutations instead of the permutations themselves. 

The main technical ingredient in this section is to show that if the the labels can be grouped into collections in such a way that the probability that a vertex has a child in a given collection depends only on what collection that vertex is in, then we can replace the labels with the collections without making it easier to determine $X^{(0)}$ with nontrivial accuracy. More formally, we have the following.

\begin{lemma}\label{lem:sym}
Consider a generalized broadcast tree model with parameters $m$, $k$ and $M$. Suppose there is a partition $S_1,\cdots,S_{m'}$ of $\{1,\cdots,m\}$ with the following property: Let $w^{(i)}=\sum_{j\in S_i} e_j$ for each $i$. Then for all $1\le i,i'\le m'$ and $j,j'\in S_i$ we have $$w^{(i')}\cdot M e_j= w^{(i')}\cdot M e_{j'}$$ Finally let $M'$ be the $m'\times m'$ matrix such that for each $i$, $i'$, $M'_{i,i'}= w^{(i)}\cdot M e_j$ for some $j\in S_{i'}$. If there is a $\mathbf{TC}^0$ detection function for the generalized broadcast process with parameters $(m',M')$ then there is a $\mathbf{TC}^0$ detection function for $(m, M)$ as well. 
\end{lemma}

\begin{proof}
First fix any $d$ and let $(X^{(0)},\cdots,X^{(d)})$ be vectors of labels generated by the generalized broadcast process with parameters $(m,M)$. The labels of the generalized broadcast process with parameters $(m', M')$ will naturally be associated with parts of the partition. Let $(X'^{(0)},\cdots,X'^{(d)})$ be the result of replacing each label with the part it belongs to. Also, let $(X^{\star(0)},\cdots,X^{\star(d)})$ be vectors of labels generated by the generalized broadcast process with parameters $(m',M')$. We claim that the distribution of $(X'^{(0)},\cdots,X'^{(d)})$ conditioned on a fixed value of $X'^{(0)}$ is the same as that of $(X^{\star(0)},\cdots,X^{\star(d)})$ conditioned on the same value for the root label. This is because by assumption, when we only care about which part of the partition each child belongs to, it only matters what part of the partition the parent belongs to.

Now suppose that $f$ is a $\mathbf{TC}^0$ function that solves the detection problem for the generalized broadcast process with parameters $(m',M')$. Finally let $\hat{X}$ be a random label contained in $S_{f(X'^{(d)})}$. Then
\begin{align*}
\mathbb{P}[\hat{X}=X^{(0)}]&=\sum_{i=1}^{m'} \mathbb{P}[\hat{X}=X^{(0)}|X'^{(0)}=i]\mathbb{P}[X'^{(0)}=i]\\
&=\sum_{i=1}^{m'} \frac{\mathbb{P}[f(X'^{(d)})=X'^{(0)}|X'^{(0)}=i]}{|S_i|}\cdot\frac{|S_i|}{m}\\
&=\frac{1}{m}\sum_{i=1}^{m'} \mathbb{P}[f(X^{\star(d)})=X^{\star(0)}|X^{\star(0)}=i] =\frac{m'}{m} \mathbb{P}[f(X^{\star(d)})=X^{\star(0)}]
\end{align*}
Also note that we can compute $X'^{(d)}$ from $X^{(d)}$ using an $\mathbf{NC}^0$ circuit. Putting it all together, because $f\in \mathbf{TC}^0$, and there must be a specific way to choose a value of $\hat{X}$ for each possible value of $f(X'^{(d)})$ such that $$\mathbb{P}[\hat{X}=X^{(0)}]\ge \frac{m'}{m} \mathbb{P}[f(X^{\star(d)})=X^{\star(0)}]$$ Hence there is a $\mathbf{TC}^0$ circuit that computes $X^{(0)}$ from $X^{(d)}$ with nontrivial advantage.
\end{proof}

In particular, we can now reduce the number of labels in our generalized broadcast process as follows. We will call this final model the generalized broadcast process on conjugacy classes.  There are $16$ labels corresponding to the ordered pairs of conjugacy classes of even permutations in $S_5$. In order to assign a label to a vertex's child, first let $\sigma$ be a random element of the vertex's label's first conjugacy class with probability $2/3$ and a random element of its label's second conjugacy class with probability $1/3$. Then, select random $\sigma',\sigma''\in A_5$ such that $\sigma'\cdot\sigma''=\sigma$. Finally, set the child's label equal to the pair of the conjugacy classes of $\sigma'$ and $\sigma''$.

\begin{theorem}\label{thm:nc1main2}
 If there is an $\mathbf{TC}^0$ detection function for the generalized broadcast process on conjugacy classes then $\mathbf{TC}^0=\mathbf{NC}^1$.
\end{theorem}

\begin{proof}
What we need to do is verify that partitioning $A_5$ into conjugacy classes in $S_5$ satisfies the conditions in Lemma~\ref{lem:sym}. 
First, observe that given any $\sigma,\overline{\sigma}$ in the same conjugacy class of $S_5$, there exists $c\in S_5$ such that $\overline{\sigma}=c^{-1}\sigma c$. So, if $\sigma'\cdot\sigma''=\sigma$ then $(c^{-1} \sigma'c)\cdot(c^{-1} \sigma''c)=\overline{\sigma}$. That gives us a bijection between pairs of permutations in any given pair of conjugacy classes with a product of $\sigma$ and pairs of permutations in that pair of conjugacy classes with a product of $\overline{\sigma}$. So, if we set $S_1,\cdots,S_{16}$ equal to the sets of pairs of permutations in each even conjugacy class of $S_5$ then by Lemma~\ref{lem:sym} we have that if $\mathbf{TC}^0\neq\mathbf{NC}^1$ there is no $\mathbf{TC}^0$ detection function for this instance of the generalized broadcast process on conjugacy classes.
\end{proof}

Finally we show that the detection problem can be solved in $\mathbf{NC}^1$ where as before we set $k = 60000$. First we note that one of the conjugacy classes contains only the identity, so if a vertex's label is $(S,S')$, then each of its children have a label of $(\{1\},S)$ with probability $1/90$, $(\{1\},S')$ with probability $1/180$, and no other possibility of having a label with its first entry equal to the identity's conjugacy class. As such, if we can determine the labels of the vertices at depth $d'$ with accuracy $.999$ then for each vertex at depth $d'-1$, we can estimate how many children it has with label $(\{1\},S'')$ for each conjugacy class $S''$ and use that to determine its label with accuracy at least $.999$. Therefore, by induction we can determine the root's label correctly with probability at least $.999$. Also, this can clearly be done in $\mathbf{NC}^1$.

\end{document}